\documentclass[aps,pra,twocolumn,showpacs,superscriptaddress,superscriptaddress,amsthm,amsmath,amssymb]{revtex4-1}  
\usepackage[lofdepth,lotdepth]{subfig}
\usepackage{graphicx}% Include figure files
\usepackage{dcolumn}% Align table columns on decimal point
\usepackage{bm}
\usepackage{multirow}
\usepackage{amsthm} % theorem/definition/lemma environments
\usepackage{pgfplots}
\usepackage{amsmath}
\usepackage{lmodern}
\usepackage{pifont}
\pgfplotsset{compat=newest}
\usepackage{float}

\usepackage{xy}
\xyoption{matrix}
\xyoption{frame}
\xyoption{arrow}
\xyoption{arc}

\usepackage{ifpdf}
\ifpdf
\else
\PackageWarningNoLine{Qcircuit}{Qcircuit is loading in Postscript mode.  The Xy-pic options ps and dvips will be loaded.  If you wish to use other Postscript drivers for Xy-pic, you must modify the code in Qcircuit.tex}
\xyoption{ps}
\xyoption{dvips}
\fi

\entrymodifiers={!C\entrybox}

\newcommand{\ket}[1]{{\left\vert{#1}\right\rangle}}
\newcommand{\qw}[1][-1]{\ar @{-} [0,#1]}
\newcommand{\qwx}[1][-1]{\ar @{-} [#1,0]}
\newcommand{\gate}[1]{*+<.6em>{#1} \POS ="i","i"+UR;"i"+UL **\dir{-};"i"+DL **\dir{-};"i"+DR **\dir{-};"i"+UR **\dir{-},"i" \qw}
\newcommand{\control}{*!<0em,.025em>-=-<.2em>{\bullet}}
\newcommand{\controlone}{*!<0em,.025em>-=-<.025em>{\text{\ding{182}}}}
\newcommand{\controltwo}{*!<0em,.025em>-=-<.025em>{\text{\ding{183}}}}
\newcommand{\controlthree}{*!<0em,.025em>-=-<.025em>{\text{\ding{184}}}}
\newcommand{\controlfour}{*!<0em,.025em>-=-<.025em>{\text{\ding{185}}}}
\newcommand{\controlfive}{*!<0em,.025em>-=-<.025em>{\text{\ding{186}}}}
\newcommand{\controlsix}{*!<0em,.025em>-=-<.025em>{\text{\ding{187}}}}

\newcommand{\controloone}{*!<0em,.025em>-=-<.025em>{\text{\ding{172}}}}
\newcommand{\controlotwo}{*!<0em,.025em>-=-<.025em>{\text{\ding{173}}}}
\newcommand{\controlothree}{*!<0em,.025em>-=-<.025em>{\text{\ding{174}}}}
\newcommand{\controlofour}{*!<0em,.025em>-=-<.025em>{\text{\ding{175}}}}
\newcommand{\controlofive}{*!<0em,.025em>-=-<.025em>{\text{\ding{176}}}}
\newcommand{\controlosix}{*!<0em,.025em>-=-<.025em>{\text{\ding{177}}}}

\newcommand{\ctrl}[1]{\control \qwx[#1] \qw}
\newcommand{\ctrlone}[1]{\controlone \qwx[#1] \qw}
\newcommand{\ctrltwo}[1]{\controltwo \qwx[#1] \qw}
\newcommand{\ctrlthree}[1]{\controlthree \qwx[#1] \qw}
\newcommand{\ctrlfour}[1]{\controlfour \qwx[#1] \qw}
\newcommand{\ctrlfive}[1]{\controlfive \qwx[#1] \qw}
\newcommand{\ctrlsix}[1]{\controlsix \qwx[#1] \qw}

\newcommand{\ctrloone}[1]{\controloone \qwx[#1] \qw}
\newcommand{\ctrlotwo}[1]{\controlotwo \qwx[#1] \qw}
\newcommand{\ctrlothree}[1]{\controlothree \qwx[#1] \qw}
\newcommand{\ctrlofour}[1]{\controlofour \qwx[#1] \qw}
\newcommand{\ctrlofive}[1]{\controlofive \qwx[#1] \qw}
\newcommand{\ctrlosix}[1]{\controlosix \qwx[#1] \qw}

\newcommand{\targ}{*+<.02em,.02em>{\xy ="i","i"-<.39em,0em>;"i"+<.39em,0em> **\dir{-}, "i"-<0em,.39em>;"i"+<0em,.39em> **\dir{-},"i"*\xycircle<.4em>{} \endxy} \qw}

\newcommand{\targone}{*+<.02em,.02em>{\xy ="i","i"-<.49em,.0em>;"i"-<.31em,0.0em> **\dir{-}, "i"+<.31em,.0em>;"i"+<.49em,0.0em> **\dir{-}, "i"-<0em,.49em>;"i"-<0em,.31em> **\dir{-}, "i"+<0em,.31em>;"i"+<0em,.49em> **\dir{-}, "i"*\xycircle<.5em>{},"i"*{\text{\ding{182}}} \endxy} \qw}
\newcommand{\targtwo}{*+<.02em,.02em>{\xy ="i","i"-<.49em,.0em>;"i"-<.31em,0.0em> **\dir{-}, "i"+<.31em,.0em>;"i"+<.49em,0.0em> **\dir{-}, "i"-<0em,.49em>;"i"-<0em,.31em> **\dir{-}, "i"+<0em,.31em>;"i"+<0em,.49em> **\dir{-}, "i"*\xycircle<.5em>{},"i"*{\text{\ding{183}}} \endxy} \qw}
\newcommand{\targthree}{*+<.02em,.02em>{\xy ="i","i"-<.49em,.0em>;"i"-<.31em,0.0em> **\dir{-}, "i"+<.31em,.0em>;"i"+<.49em,0.0em> **\dir{-}, "i"-<0em,.49em>;"i"-<0em,.31em> **\dir{-}, "i"+<0em,.31em>;"i"+<0em,.49em> **\dir{-}, "i"*\xycircle<.5em>{},"i"*{\text{\ding{184}}} \endxy} \qw}
\newcommand{\targfour}{*+<.02em,.02em>{\xy ="i","i"-<.49em,.0em>;"i"-<.31em,0.0em> **\dir{-}, "i"+<.31em,.0em>;"i"+<.49em,0.0em> **\dir{-}, "i"-<0em,.49em>;"i"-<0em,.31em> **\dir{-}, "i"+<0em,.31em>;"i"+<0em,.49em> **\dir{-}, "i"*\xycircle<.5em>{},"i"*{\text{\ding{185}}} \endxy} \qw}
\newcommand{\targfive}{*+<.02em,.02em>{\xy ="i","i"-<.49em,.0em>;"i"-<.31em,0.0em> **\dir{-}, "i"+<.31em,.0em>;"i"+<.49em,0.0em> **\dir{-}, "i"-<0em,.49em>;"i"-<0em,.31em> **\dir{-}, "i"+<0em,.31em>;"i"+<0em,.49em> **\dir{-}, "i"*\xycircle<.5em>{},"i"*{\text{\ding{186}}} \endxy} \qw}
\newcommand{\targsix}{*+<.02em,.02em>{\xy ="i","i"-<.49em,.0em>;"i"-<.31em,0.0em> **\dir{-}, "i"+<.31em,.0em>;"i"+<.49em,0.0em> **\dir{-}, "i"-<0em,.49em>;"i"-<0em,.31em> **\dir{-}, "i"+<0em,.31em>;"i"+<0em,.49em> **\dir{-}, "i"*\xycircle<.5em>{},"i"*{\text{\ding{187}}} \endxy} \qw}
\newcommand{\targoone}{*+<.02em,.02em>{\xy ="i","i"-<.49em,.0em>;"i"-<.34em,0.0em> **\dir{-}, "i"+<.34em,.0em>;"i"+<.49em,0.0em> **\dir{-}, "i"-<0em,.49em>;"i"-<0em,.34em> **\dir{-}, "i"+<0em,.34em>;"i"+<0em,.49em> **\dir{-}, "i"*\xycircle<.5em>{},"i"*{\text{\ding{172}}} \endxy} \qw}
\newcommand{\targotwo}{*+<.02em,.02em>{\xy ="i","i"-<.49em,.0em>;"i"-<.34em,0.0em> **\dir{-}, "i"+<.34em,.0em>;"i"+<.49em,0.0em> **\dir{-}, "i"-<0em,.49em>;"i"-<0em,.34em> **\dir{-}, "i"+<0em,.34em>;"i"+<0em,.49em> **\dir{-}, "i"*\xycircle<.5em>{},"i"*{\text{\ding{173}}} \endxy} \qw}
\newcommand{\targothree}{*+<.02em,.02em>{\xy ="i","i"-<.49em,.0em>;"i"-<.34em,0.0em> **\dir{-}, "i"+<.34em,.0em>;"i"+<.49em,0.0em> **\dir{-}, "i"-<0em,.49em>;"i"-<0em,.34em> **\dir{-}, "i"+<0em,.34em>;"i"+<0em,.49em> **\dir{-}, "i"*\xycircle<.5em>{},"i"*{\text{\ding{174}}} \endxy} \qw}
\newcommand{\targofour}{*+<.02em,.02em>{\xy ="i","i"-<.49em,.0em>;"i"-<.34em,0.0em> **\dir{-}, "i"+<.34em,.0em>;"i"+<.49em,0.0em> **\dir{-}, "i"-<0em,.49em>;"i"-<0em,.34em> **\dir{-}, "i"+<0em,.34em>;"i"+<0em,.49em> **\dir{-}, "i"*\xycircle<.5em>{},"i"*{\text{\ding{175}}} \endxy} \qw}
\newcommand{\targofive}{*+<.02em,.02em>{\xy ="i","i"-<.49em,.0em>;"i"-<.34em,0.0em> **\dir{-}, "i"+<.34em,.0em>;"i"+<.49em,0.0em> **\dir{-}, "i"-<0em,.49em>;"i"-<0em,.34em> **\dir{-}, "i"+<0em,.34em>;"i"+<0em,.49em> **\dir{-}, "i"*\xycircle<.5em>{},"i"*{\text{\ding{176}}} \endxy} \qw}
\newcommand{\targosix}{*+<.02em,.02em>{\xy ="i","i"-<.49em,.0em>;"i"-<.34em,0.0em> **\dir{-}, "i"+<.34em,.0em>;"i"+<.49em,0.0em> **\dir{-}, "i"-<0em,.49em>;"i"-<0em,.34em> **\dir{-}, "i"+<0em,.34em>;"i"+<0em,.49em> **\dir{-}, "i"*\xycircle<.5em>{},"i"*{\text{\ding{177}}} \endxy} \qw}
\newcommand{\multigate}[2]{*+<1em,.9em>{\hphantom{#2}} \POS [0,0]="i",[0,0].[#1,0]="e",!C *{#2},"e"+UR;"e"+UL **\dir{-};"e"+DL **\dir{-};"e"+DR **\dir{-};"e"+UR **\dir{-},"i" \qw}
\newcommand{\ghost}[1]{*+<1em,.9em>{\hphantom{#1}} \qw}
\newcommand{\gategroup}[6]{\POS"#1,#2"."#3,#2"."#1,#4"."#3,#4"!C*+<#5>\frm{#6}}
\newcommand{\rstick}[1]{*!L!<-.5em,0em>=<0em>{#1}}
\newcommand{\lstick}[1]{*!R!<.5em,0em>=<0em>{#1}}
\newcommand{\Qcircuit}{\xymatrix @*=<0em>}
\theoremstyle{plain}

\theoremstyle{definition}
\newtheorem{defin}{Definition}
\newenvironment{defi}{\vspace{0mm}\begin{defin}}{\end{defin}}
\theoremstyle{remark}
\newtheorem{col}{Corollary}
\newtheorem{lemm}{Proposition}

\renewenvironment{proof}{\noindent{\bf Proof:}\;}{$\square$\,}

\begin{document}

\title{On the advantages of using relative phase Toffolis with an application to multiple control Toffoli optimization}
\author{Dmitri Maslov}
\email{mailto:dmitri.maslov@gmail.com}
\altaffiliation{Joint Center for Quantum Information and Computer Science, University of Maryland, College Park, MD, USA}
\affiliation{National Science Foundation, Arlington, Virginia, USA}

\begin{abstract}
Various implementations of the Toffoli gate up to a relative phase have been known for years.  The advantage over regular Toffoli gate is their smaller circuit size.  However, their use has been often limited to a demonstration of quantum control in designs such as those where the Toffoli gate is being applied last or otherwise for some specific reasons the relative phase does not matter.  It was commonly believed that the relative phase deviations would prevent the relative phase Toffolis from being very helpful in practical large-scale designs. 

In this paper, we report three circuit identities that provide the means for replacing certain configurations of the multiple control Toffoli gates with their simpler relative phase implementations, up to a selectable unitary on certain qubits, and without changing the overall functionality.  We illustrate the advantage via applying those identities to the optimization of the known circuits implementing multiple control Toffoli gates, and report the reductions in the CNOT-count, $T$-count, as well as the number of ancillae used.  We suggest that a further study of the relative phase Toffoli implementations and their use may yield other optimizations.
\end{abstract}


\maketitle

\section{Introduction} 

Multiple control Toffoli gates are the staple of quantum arithmetic and reversible circuits.  They are employed widely within quantum algorithms, including in reversible transformations, such as arithmetic circuits and all sorts of Boolean operations over quantum registers, as well as subroutines within other specialized quantum transforms.  Unfortunately, multiple control Toffoli gates are not simple operations, and require to be implemented using a certain library of elementary gates---physically attainable transformations for physical-level implementations, and fault-tolerant gates on the logical level.  As of the time of this writing, most advanced and developed trapped ions \cite{ar:mk} and superconducting \cite{ar:dh} quantum information processing approaches allow computations over at most a few dozen qubits using at most a few dozen two-qubit gates.  The smallest of the multiple control Toffoli gates, the three-qubit Toffoli gate, requires six CNOT gates as a physical-level circuit over controlling apparatus allowing the application of the CNOT and arbitrary single qubit gates, and seven $T$ gates, as a logical fault-tolerant circuit over Clifford+$T$ library without ancillae.  The known implementations of larger multiple control Toffoli gates come at a substantially higher cost.  This makes the multiple control Toffoli gates be expensive computing primitives.  As such, the ability to replace them with their simpler counterparts that nevertheless can guarantee the overall functional integrity, as well as their optimization (multiple control Toffoli gates are implemented using smaller size multiple control Toffolis \cite{j:bbcd,bk:nc}) are important in practice.  Ultimately, the difficulty of implementing Toffoli gates may even be a deciding factor in the ability to run an experiment of a desired size.  Indeed, consider a scenario where only a fixed number of certain elementary gates can be applied.  Imagine the goal is to run a discrete logarithm type computation \cite{bk:nc}.  Since circuits implementing such an algorithm are dominated by reversible arithmetic operations, which in turn rely on the Toffoli gates, it is conceivable that optimizing Toffoli implementations would yield a resource count that is possible to execute for a desired size computation.  Multiple control Toffoli gates are, of course, important beyond just the discrete logarithm type algorithms. 

The goal of this paper is to provide a framework for replacing multiple control Toffoli gates with their simpler relative phase implementations.  The advantage is illustrated through an optimization of the implementations of the multiple control Toffoli gates.  The reported optimization is viewed as a motivating example rather than a complete and finished study.  An in-depth look at the implementations of the relative phase multiple control Toffoli gates and their use in the optimization of arbitrary quantum circuits may likely yield more results.  

To draw a classical analogy, relative phase Toffoli gates may turn out to play a role analogous to the classical NAND gates: while classical (quantum/reversible) circuits are designed using a convenient for a human set of operations (multiple control Toffolis), a compiler may decompose those into NAND gates (relative phase multiple control Toffolis) before they are mapped into lowest-level transistors (elementary quantum gates).  

\section{Definitions}

In this paper, we will work with pure $n$-qubit quantum states $\sum\limits_{i=0}^{2^n-1} \alpha_{i}\ket{i}$ and quantum transformations described by the $2^n \times 2^n$ unitary matrices $U$.  Recall that a square matrix $U$ is called unitary if its inverse equals to its conjugate transpose, $U^{-1}=U^\dagger$.  While the property of unitarity defines evolutions that are possible to attain physically, it does not prescribe which ones may be implemented directly.  To assist with the presentation of the material, we will discretize the family of transformations that may be obtained physically, and call them elementary quantum gates.  This does not limit the applicability of the results---indeed, discrete circuits may be thought of as certain versions of continuous Hamiltonians, but are otherwise easier to work with.  In particular, in this work we will rely on the following elementary gates: Pauli-X, $X=NOT=\left(\begin{array}{cc} 0 & 1 \\ 1 & 0 \end{array} \right)$, Pauli-$Z$, $Z=\left(\begin{array}{cc} 1 & 0 \\ 0 & -1 \end{array} \right)$, and its roots Phase, $P=\sqrt{Z}=\left(\begin{array}{cc} 1 & 0 \\ 0 & i \end{array} \right)$, $T=\sqrt[4]{Z}=\left(\begin{array}{cc} 1 & 0 \\ 0 & \frac{1+i}{\sqrt{2}} \end{array} \right)$, and Pauli-$Y$, $Y=\left(\begin{array}{cc} 0 & -i \\ i & 0 \end{array} \right)$.  A fourth root of $Y$ will be mentioned in some constructions, in the form of $R_Y(\pi/4)$, that is equivalent to the fourth root of $Y$ up to a global phase.  Recall that $R_Y(\theta)=\left(\begin{array}{cc} \cos\frac{\theta}{2} & -\sin\frac{\theta}{2} \\ \sin\frac{\theta}{2} & \cos\frac{\theta}{2} \end{array} \right)$. Finally, for completeness we will need the Hadamard gate, $H=\frac{1}{\sqrt{2}}\left(\begin{array}{cc} 1 & 1 \\ 1 & -1 \end{array} \right)$, and the two-qubit CNOT gate that we introduce via the mapping of kets, rather than the $4 \times 4$ matrix, as CNOT$(a,b): \ket{a,b} \mapsto \ket{a,b\oplus a}$, and everywhere else by linearity, due to the simplicity of such a definition.  

Quantum circuits are defined as the strings of quantum gates, or otherwise products of matrices that correspond to the individual gates.  For multiple qubit circuit computations via matrices, a proper Kronecker product needs to be taken to compute matrix products.  For example, a two-qubit operation corresponding to the Hadamard gate on the first qubit is given by the matrix $H \otimes Id$, where $Id$ is the identity applied to the second qubit.  Recall that the product of matrices is taken in reverse order with respect to the order of gates in the corresponding circuit.  Following the standard notations, circuits/unitaries composed of quantum gates/matrices $X$, $Y$, $Z$, $P$, $H$, and CNOT are called Clifford.  These unitaries play an important role in quantum error correction, but are not complete (moreover, simulable classically with a polynomial size effort) for quantum computation.  As such, for completeness, a circuit library needs to contain a non-Clifford gate, such as the $T$ gate.  The addition of any non-Clifford gate to the Clifford circuits furthermore turns out to result in the computational universality \cite{bk:nc}.

The above is meant to be a quick reminder of some basic facts and an introduction of the notations used in this paper.  For an in-depth review we refer the reader to \cite{bk:nc}.

For convenience, we furthermore use the following notations: for a set of variables/qubits $X=\{x_1,x_2,...,x_n\}$, $|X|$ equals $n$, being the number of individual qubits in this set, and the conjugation (Boolean AND) of variables, $x_1 \& x_2 \& ... \& x_n$ is denoted as simply $x$.  When the number of variables in the set $X$ is zero, we assign $x$ the value of 1.  When the set of variables $X$ consists of a single element, $\{x\}$, the conjugation of the variables within the set, as well as the name of the variable, coincide; this does not however cause any issues.  

We next define the multiple control Toffoli gates. 

\begin{defi}
A {\em multiple control Toffoli gate} over a set of $n$ qubits with the set $X=\{x_1,x_2,...,x_{n-1}\}$ being the controls, and qubit $y$ being the target, $TOF^n(X;y)$, is defined as the matrix $$diag\left\{1,1,....,1,\left(\begin{array}{cc} 0 & 1 \\ 1 & 0 \end{array} \right)\right\}.$$
\end{defi}

We will sometimes omit the superscript and write $TOF(X;y)$ when the controls and the target are explicitly specified and the size of the multiple control Toffoli gate can thus be restored.  Similarly, we may omit the specification of the qubits the gate operates on and write $TOF^n$ when we are only concerned with the size of the gate.  Finally, we may write $TOF$ when the goal is to specify the kind of gate being the Toffoli and distinguish it from other kinds of gates.  Observe, that when $|X|=0$ the above definition reports the Pauli-$X$ ($NOT$) gate, for $|X|=1$ the definition introduces the CNOT gate, when $|X|=2$, it reduces to the usual Toffoli gate $TOF^3$, and for larger sets $X$, the multiply-controlled Toffolis---Toffoli-4, Toffoli-5, {\em etc.}  

An alternate definition of the multiple control Toffoli gate may cast it in the form of the mapping of kets, as follows, $TOF^n(X;y): \ket{X;y} \mapsto \ket{X, y \oplus x}$.  In some cases, the mapping of kets may be easier to operate with than the corresponding unitary matrix.

In our constructions, relative phase implementations of quantum unitary transformations play a major role.  For the purpose of this work, we define relative phase implementations as follows. 

\begin{defi}
A {\em relative phase} version of a quantum $n$-qubit unitary operation $U=\{u_{i,j}\}|_{i,j = 0..2^n-1}$ is any $n$-qubit unitary $V=\{v_{i,j}\}|_{i,j = 0..2^n-1}$ such that $|v_{i,j}|=|u_{i,j}|$ for all $i$ and $j$.
\end{defi} 

In other words, a relative phase version or otherwise implementation of a unitary $U$ is a unitary $V$ such that the elements of the two matrices differ by $e^{i\pi \phi}$, where $\phi \in \mathbb{R}$, and $\phi$ may be different for different matrix elements.  Observe that $e^{i\pi \phi}0=0$, therefore relative phase versions of unitaries have zeroes everywhere the original unitary does. 

To illustrate, a relative phase multiple control Toffoli gate over the set of controls $X=\{x_1,x_2,...,x_{n-1}\}$ with the target $y$, $RTOF(X;y)$, can be written as follows,  
$$diag\left\{z_0,z_1,....,z_{2^{n}-3},\left(\begin{array}{cc} 0 & z_{2^{n}-2} \\ z_{2^{n}-1} & 0 \end{array} \right)\right\},$$
where $z_i$ are arbitrary length-1 complex numbers. Prefix ``$R$'' is used to distinguish the relative phase version from the multiple control Toffoli gate itself.  Observe that when all $z_i=1$, the respective relative phase Toffoli gate $RTOF(X;y)$ becomes the multiple control Toffoli gate $TOF(X;y)$, and when all $z_i$ take the same but fixed value $z$, the respective relative phase Toffoli gate $RTOF(X;y)$ implements the multiple control Toffoli gate $TOF(X;y)$ up to an undetectable global phase $z$. 

A relative phase multiple control Toffoli gate $RTOF^n$ may be thought of as a product of the multiple control Toffoli gate $TOF^n$ and an $n$-qubit diagonal unitary $D^n$.  Indeed, for a diagonal unitary $D^n:=diag\left\{z_0,z_1,....,z_{2^{n}-1}\right\}$ circuit $TOF^nD^n$ implements a generic relative phase multiple control Toffoli gate $R_1TOF^n=diag\left\{z_0,z_1,....,z_{2^{n}-3},\left(\begin{array}{cc} 0 & z_{2^{n}-2} \\ z_{2^{n}-1} & 0 \end{array} \right)\right\}$, whereas circuit $D^nTOF^n$ implements a generic relative phase multiple control Toffoli gate $R_2TOF^n=diag\left\{z_0,z_1,....,z_{2^{n}-3},\left(\begin{array}{cc} 0 & z_{2^{n}-1} \\ z_{2^{n}-2} & 0 \end{array} \right)\right\}$.  Observe how both gates are relative phase multiple control Toffoli gates, but different in the last two non-zero elements, that are being permuted.  We will exploit this property in the circuit diagrams.  In particular, of the two possible decompositions of the relative phase multiple control Toffoli gate into a product of the multiple control Toffoli and a diagonal unitary, we will select $TOF^nD^n$ to be the canonic one, and draw the respective relative phase multiple control Toffoli gate with same controls as the diagonal gate $D^n$ and a distorted target, such as illustrated in Figure \ref{fig:1}(c).  The helpful intuition behind this pictorial representation is as follows: a Toffoli gate $TOF(X;y)$ may be combined with a diagonal gate $D(Z)$, $Z \in \{X,y\}$, following it to obtain a relative phase Toffoli gate, or a Toffoli gate $TOF(X;y)$ may be combined with a diagonal gate $D(Z)$, $Z \in \{X,y\}$, preceding it to obtain the inverse of a relative phase Toffoli gate; conversely, each relative phase Toffoli gate or its inverse may be broken down into a suitable pair of the multiple control Toffoli gate and the diagonal gate.

An important property of the relative phase multiple control Toffoli gates is that every one of those is an inverse of some other relative phase multiple control Toffoli gate. Indeed, for $R_1TOF^n=diag\left\{z_0,z_1,....,z_{2^{n}-3},\left(\begin{array}{cc} 0 & z_{2^{n}-1} \\ z_{2^{n}-2} & 0 \end{array} \right)\right\}$ and $R_2TOF^n=diag\left\{w_0,w_1,....,w_{2^{n}-3},\left(\begin{array}{cc} 0 & w_{2^{n}-1} \\ w_{2^{n}-2} & 0 \end{array} \right)\right\}$ $R_1TOF^n=R_2^{-1}TOF^n$ when $w_i=z_i^{-1}$ for $i=0...2^n-3$, $w_{2^n-2}=z_{2^n-1}^{-1}$, and $w_{2^n-1}=z_{2^n-2}^{-1}$. 

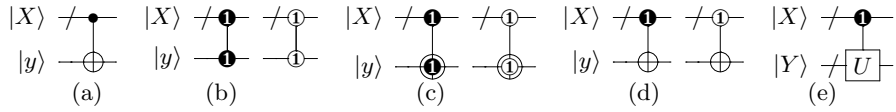
\begin{figure*}
\centerline{
\begin{tabular}{ccccc}
\Qcircuit @C=0.5em @R=1em @!R {
\lstick{\ket{X}} & {/} \qw 	& \ctrl{1}	& \qw \\
\lstick{\ket{y}} & \qw		& \targ		& \qw
}
&\hspace{8mm}
\Qcircuit @C=0.5em @R=1em @!R {
\lstick{\ket{X}} & {/} \qw 	& \ctrlone{1}	& \qw \\
\lstick{\ket{y}} & \qw		& \ctrlone{-1}	& \qw
}
\hspace{1mm}
\Qcircuit @C=0.5em @R=1em @!R {
& {/} \qw 	& \ctrloone{1}	& \qw \\
& \qw		& \ctrloone{-1}	& \qw
}
&\hspace{8mm}
\Qcircuit @C=0.5em @R=1em @!R {
\lstick{\ket{X}} & {/} \qw 	& \ctrlone{1} 	& \qw \\
\lstick{\ket{y}} & \qw		& \targone 		& \qw
} \hspace{1mm}
\Qcircuit @C=0.5em @R=1em @!R {
& {/} \qw 	& \ctrloone{1}	& \qw \\
& \qw		& \targoone		& \qw
}
&\hspace{8mm}
\Qcircuit @C=0.5em @R=1em @!R {
\lstick{\ket{X}} & {/} \qw 	& \ctrlone{1}	& \qw \\
\lstick{\ket{y}} & \qw		& \targ			& \qw
}
\hspace{1mm}
\Qcircuit @C=0.5em @R=1em @!R {
& {/} \qw 	& \ctrloone{1}	& \qw \\
& \qw		& \targ		& \qw
}
&\hspace{8mm}
\Qcircuit @C=0.5em @R=0.7em @!R {
\lstick{\ket{X}} & {/} \qw 		& \ctrlone{1}	& \qw \\
\lstick{\ket{Y}} & {/} \qw 		& \gate{U} 		& \qw
}
\\
(a)&(b)&(c)&(d)&(e)
\end{tabular}
}
\caption{(a) a multiple control Toffoli gate $TOF(X;y)$, (b) a diagonal gate $D_1(X;y)$ and its inverse; observe how different diagonal gates can be visually distinguished by the number within the control, and a diagonal gate and its inverse are related by the different color of the control, (c) a relative phase multiple control Toffoli gate $R_1TOF(X;y)$ and its inverse $R_1^{-1}TOF(X;y)$, (d) a type-$y$ special form $S^yR_1TOF(X;y)$, and its inverse, and (e) a controlled-unitary $U$ implemented up to some relative phase.  The \hspace{0.5mm} ${/} \hspace{-2.5mm}$--- symbol denotes a multiqubit register.}\label{fig:1}
\end{figure*}

We next define special form relative phase multiple control Toffoli gates, that are important in some of the constructions that follow. 

\begin{defi}
For a set $X=\{x_1,x_2,...,x_n\}$, and its subset $X^\prime=\{x_{i_1},x_{i_2},...,x_{i_k}\}$ a {\em type-$X^\prime$ special form relative phase multiple control Toffoli gate}, $S^{X^\prime}RTOF(x_1,x_2,...,x_{n-1};x_n)$, is defined as the matrix 
$$diag\left\{z_0,z_1,...,z_{2^n-3},\left(\begin{array}{cc} 0 & z_{2^n-2} \\ z_{2^n-1} & 0 \end{array} \right)\right\},$$
where every pair of complex numbers $z_s$ and $z_t$ are equal whenever the binary expansions of $s$ and $t$ are different only in the digits $i_1,i_2,...,i_{k-1},$ and $i_k$.
\end{defi}

To illustrate, a type-$\{x_1\}$ $S^{x_1}RTOF(x_1,x_2,...,x_{n-1};x_n)$ is given by the matrix
\begin{eqnarray*}
diag \{ z_0,z_1,...,z_{2^{n-1}-1}, \\ z_0,z_1,...,z_{2^{n-1}-3},\left(\begin{array}{cc} 0 & z_{2^{n-1}-2} \\ z_{2^{n-1}-1} & 0 \end{array} \right) \}.
\end{eqnarray*}

The type-$\{x_1\}$ special form relative phase Toffoli gate $S^{x_1}RTOF$ has half the number of the degrees of freedom compared to the equal size unrestricted relative phase Toffoli gate $RTOF$.  In practice, this suggests that it should be easier to find an efficient circuit implementing a relative phase Toffoli gate than it is to find one of the same size for a type-$\{x_1\}$ special form relative phase Toffoli gate.  To give another example, a type-$X$ $S^{X}RTOF(x_1,x_2,...,x_{n-1};x_n)$ is the most restrictive of the kind.  It is equal to the respective Toffoli gate up to a global phase, and thereby does not give much freedom in implementing by a circuit over the $TOF(x_1,x_2,...,x_{n-1};x_n)$.  This means that in the practical constructions, and whenever possible, we will try to use a type-$X^\prime$ special form relative phase multiple control Toffoli gate with the smallest size set $X^\prime$. 

An alternate and equivalent definition of a type-$X^\prime$ special form relative phase Toffoli gate is via a transformation given by the circuit $TOF(x_1,x_2,...,x_{n-1};x_n)D(X \setminus X^\prime)$.  It furthermore serves as a basis for how we draw $SRTOF$ gates in the circuit diagrams.  Compared to the multiple control Toffoli gate, every control/target in the set $X \setminus X^\prime$ of $S^{X^\prime}RTOF$ appears distorted by the dingbat originating from the respective $D(X \setminus X^\prime)$, and every control/target in the set $X^\prime$ appears undistorted, see Figure \ref{fig:1}(d). 

Beyond having fewer degrees of freedom compared to an unrestricted relative phase Toffoli gate, there is one more important difference between the special form relative phase Toffoli gates and the relative phase Toffoli gates: the inverse of a type-$X^\prime$ special form relative phase Toffoli gate is not always a type-$X^\prime$ special form relative phase Toffoli gate.

The use of subscripts allows to distinguish different versions of the relative phase and special form relative phase multiple control Toffoli gates.  For instance, notations $R_1TOF$ and $R_2TOF$ indicate that both gates are some relative phase Toffoli gates, but they are not necessarily related.  In contrast, an $R_1^{-1}TOF$ is the inverse of the $R_1TOF$.  Recall, that a circuit implementing the inverse operation may be constructed by conjugating the gates in the circuit implementing the given unitary and inverting their order.   Observe further that any two $TOF$ gates of the same size are represented by identical matrices; this is not always true for some two $RTOF$ or a pair of $SRTOF$, therefore the ability to distinguish different versions of the relative phase implementations is important, as these could be different gates. 

We will draw quantum gates and circuits using standard notations, including the relative phase gates per diagrams found in Figure \ref{fig:1}, with time propagating from left to right.  Some useful circuit identities clarifying and summarizing the above discussions are shown next.

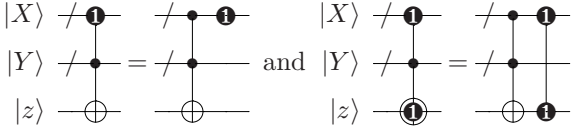
\begin{enumerate}
\item 
\[
\Qcircuit @C=0.5em @R=1em @!R {
\lstick{\ket{X}} & {/} \qw 	& \ctrlone{1}	& \qw \\
\lstick{\ket{Y}} & {/} \qw 	& \ctrl{1}		& \qw \\
\lstick{\ket{z}} & \qw		& \targ			& \qw
}
\raisebox{-1.8em}{\hspace{1mm}=\hspace{1mm}}
\Qcircuit @C=0.5em @R=1em @!R {
& {/} \qw 	& \ctrl{1}	& \ctrlone{0} 	& \qw \\
& {/} \qw 	& \ctrl{1}	& \qw			& \qw \\
& \qw		& \targ		& \qw			& \qw
}
\hspace{2mm}\raisebox{-1.8em}{\text{and}}\hspace{9mm}
\Qcircuit @C=0.5em @R=0.8em @!R {
\lstick{\ket{X}} & {/} \qw 	& \ctrlone{1}	& \qw \\
\lstick{\ket{Y}} & {/} \qw 	& \ctrl{1}		& \qw \\
\lstick{\ket{z}} & \qw		& \targone		& \qw
}
\raisebox{-1.8em}{\hspace{1mm}=\hspace{1mm}}
\Qcircuit @C=0.5em @R=1em @!R {
& {/} \qw 	& \ctrl{1}	& \ctrlone{2} 	& \qw \\
& {/} \qw 	& \ctrl{1}	& \qw			& \qw \\
& \qw		& \targ		& \ctrlone{0}	& \qw
}
\]
show canonic decomposition of $S^{Y,z}R_1TOF$ and $S^{Y}R_1TOF$ into the product of the multiple control Toffoli gate $TOF$ and the diagonal gate $D_1$; read right-to-left, these rules show how to combine a suitable pair of the multiple control Toffoli gate and the diagonal gate into a (special form) relative phase Toffoli gate.  When $Y = \emptyset$, second circuit illustrates the $R_1TOF$ gate.
\item 
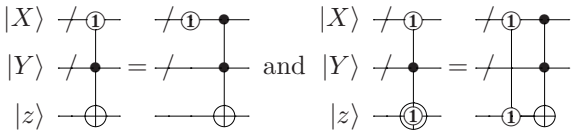
\[
\Qcircuit @C=0.5em @R=1em @!R {
\lstick{\ket{X}} & {/} \qw 	& \ctrloone{1}	& \qw \\
\lstick{\ket{Y}} & {/} \qw 	& \ctrl{1}		& \qw \\
\lstick{\ket{z}} & \qw		& \targ			& \qw
}
\raisebox{-1.8em}{\hspace{1mm}=\hspace{1mm}}
\Qcircuit @C=0.5em @R=1em @!R {
& {/} \qw 	& \ctrloone{0} 	& \ctrl{1}	& \qw \\
& {/} \qw 	& \qw			& \ctrl{1}	& \qw \\
& \qw		& \qw			& \targ		& \qw
}
\hspace{2mm}\raisebox{-1.8em}{\text{and}}\hspace{9mm}
\Qcircuit @C=0.5em @R=0.8em @!R {
\lstick{\ket{X}} & {/} \qw 	& \ctrloone{1}	& \qw \\
\lstick{\ket{Y}} & {/} \qw 	& \ctrl{1}		& \qw \\
\lstick{\ket{z}} & \qw		& \targoone		& \qw
}
\raisebox{-1.8em}{\hspace{1mm}=\hspace{1mm}}
\Qcircuit @C=0.5em @R=1em @!R {
& {/} \qw 	& \ctrloone{2} 	& \ctrl{1}	& \qw \\
& {/} \qw 	& \qw			& \ctrl{1}	&\qw \\
& \qw		& \ctrloone{0}	& \targ		\qw
}
\]
show canonic decomposition of $S^{Y,z}R_1^{-1}TOF$ and $S^{Y}R_1^{-1}TOF$ into the product of the diagonal gate and the multiple control Toffoli gate.  Indeed, looking at the second of the two identities, 
\begin{eqnarray*}
S^yR^{-1}TOF(X,Y;z) \\ = \left( TOF(X,Y;z)D_1(X,z)\right)^{-1} \\ = D_1^{-1}(X,z) TOF(X,Y;z),
\end{eqnarray*} 
being the circuit pictured on the right hand side.
\item $\forall$\ding{182} $\exists$\ding{173}:
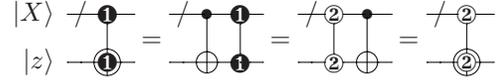
\[
\Qcircuit @C=0.5em @R=0.8em @!R {
\lstick{\ket{X}} & {/} \qw 	& \ctrlone{1}	& \qw \\
\lstick{\ket{z}} & \qw		& \targone		& \qw
}
\raisebox{-1em}{\hspace{1mm}=\hspace{1mm}}
\Qcircuit @C=0.5em @R=1em @!R {
& {/} \qw 	& \ctrl{1}	& \ctrlone{1} &\qw \\
& \qw		& \targ		& \ctrlone{0} &\qw
}
\raisebox{-1em}{\hspace{1mm}=\hspace{1mm}}
\Qcircuit @C=0.5em @R=1em @!R {
& {/} \qw 	& \ctrlotwo{1} & \ctrl{1}	& \qw \\
& \qw		& \ctrlotwo{0} & \targ		& \qw
}
\raisebox{-1em}{\hspace{1mm}=\hspace{1mm}}
\Qcircuit @C=0.5em @R=0.8em @!R {
& {/} \qw 	& \ctrlotwo{1}	& \qw \\
& \qw		& \targotwo		& \qw
}
\]
in other words, every $R_1TOF$ is also an $R_2^{-1}TOF$ under the proper choice of relative phases.
\end{enumerate} 

In general, for any reversible gate $R(X)$ its relative phase version could be thought of as a product $R(X)D(X)$, for a proper diagonal unitary $D(X)$.  This suggests a possible route in which the work reported in this paper may be extended.

\section{Main result}

Our main result is summarized in the next three Propositions.  We apply it to obtain multiple Corollaries, and to optimize multiple control Toffoli gates in the section that follows.  The proofs of these three propositions rely on the three circuit identities concluding previous section, as well as the following notion: the controlled-$U$ implemented up to a relative phase, $RCU(V,W;X)$, commutes with the controlled-$V$ implemented up to a relative phase, $RCV(V,Y;Z)$, where the qubit sets $V, W, X, Y,$ and $Z$ are disjoint.  This rule also applies to show that any two non-intersecting unitaries commute.  We assume reader's familiarity with the above commutation rule, and do not explicitly prove it here.

\begin{lemm}\label{lemma:main1}
The conjugation of the controlled unitary $U$ over the qubit set $Z$ implemented up to a possible relative phase, $R_1CU(Y,a;Z)$, by a pair of multiple control Toffoli gates $TOF(X;a)$ allows the replacement of these multiple control Toffoli gates with their relative phase versions implemented up to any desired unitary $V(X)$, such as illustrated next:
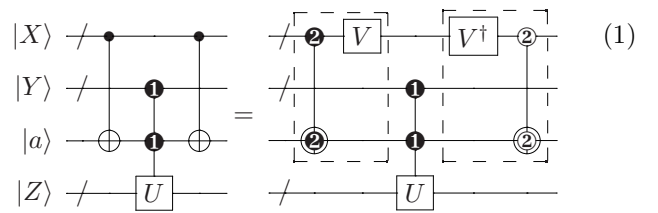
\begin{eqnarray}\label{circ:main1}
\Qcircuit @C=0.6em @R=0.7em @!R {
\lstick{\ket{X}} & {/} \qw 	& \ctrl{2}	& \qw	 		& \ctrl{2}	& \qw \\
\lstick{\ket{Y}} & {/} \qw	& \qw		& \ctrlone{1} 	& \qw		& \qw \\
\lstick{\ket{a}} & \qw		& \targ		& \ctrlone{1} 	& \targ		& \qw \\
\lstick{\ket{Z}} & {/} \qw	& \qw 		& \gate{U} 		& \qw	 	& \qw
}
&
\raisebox{-3.0em}{\hspace{1mm}=\hspace{2mm}}
&
\Qcircuit @C=0.6em @R=.53em @!R {
& {/} \qw 	& \ctrltwo{2}	& \gate{V} 	& \qw		 	& \gate{V^\dagger}	& \ctrlotwo{2}		& \qw \\
& {/} \qw	& \qw			& \qw		& \ctrlone{1} 	& \qw 				& \qw				& \qw \\
& \qw		& \targtwo		& \qw		& \ctrlone{1} 	& \qw 				& \targotwo			& \qw \\
& {/} \qw	& \qw 			& \qw		& \gate{U} 		& \qw 				& \qw	 			& \qw 
\gategroup{1}{3}{3}{4}{0.5em}{--} \gategroup{1}{6}{3}{7}{0.5em}{--}
}
\end{eqnarray}
\end{lemm}
\begin{proof}
The proof is accomplished via the following set of circuit transformations: 
\begin{eqnarray*}
TOF(X;a)R_1CU(Y,a;Z)TOF(X;a) \\
= TOF(X;a)D_2(X;a)D_2^{-1}(X;a) \\ R_1CU(Y,a;Z) TOF(X;a) \\
= TOF(X;a)D_2(X;a)R_1CU(Y,a;Z) \\ D_2^{-1}(X;a) TOF(X;a) \\
= R_2TOF(X;a)R_1CU(Y,a;Z)R_2^{-1}TOF(X;a) \\
= R_2TOF(X;a)V(X)V^{-1}(X)R_1CU(Y,a;Z) \\ R_2^{-1}TOF(X;a) \\
= \big[ R_2TOF(X;a)V(X) \big] R_1CU(Y,a;Z) \\ \big[ V^{-1}(X)R_2^{-1}TOF(X;a) \big].
\end{eqnarray*}

\end{proof} 

The result of Proposition \ref{lemma:main1} can be reduced to the following form once $RCU(Y,a;Z)$ is set to implement the Toffoli type gate, $TOF(Y,a;z)$:
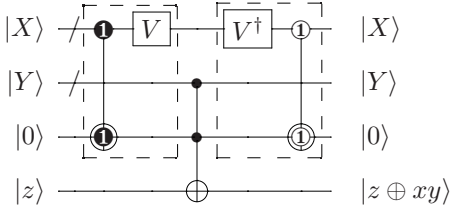
\begin{eqnarray*}
\Qcircuit @C=0.6em @R=.6em @!R {
\lstick{\ket{X}} & {/} \qw 	& \ctrlone{2}	& \gate{V} 	& \qw	 	& \gate{V^\dagger}	& \ctrloone{2} 	& \qw & \rstick{\ket{X}}\\
\lstick{\ket{Y}} & {/} \qw	& \qw			& \qw		& \ctrl{1} 	& \qw 				& \qw 			& \qw & \rstick{\ket{Y}}\\
\lstick{\ket{0}} & \qw		& \targone		& \qw		& \ctrl{1} 	& \qw 				& \targoone	 	& \qw & \rstick{\ket{0}}\\
\lstick{\ket{z}} & \qw		& \qw 			& \qw		& \targ 	& \qw 				& \qw 			& \qw & \rstick{\ket{z \oplus xy}}
\gategroup{1}{3}{3}{4}{0.5em}{--} \gategroup{1}{6}{3}{7}{0.5em}{--}
}
\end{eqnarray*}
 
Indeed, the corresponding circuit on the left hand side in (\ref{circ:main1}) computes
\begin{eqnarray*}
\ket{X,Y,0,z} \overset{TOF(X;0)}{\mapsto} \ket{X,Y,x,z} \overset{TOF(Y,x;z)}{\mapsto} \\ \ket{X,Y,x,z \oplus xy} \overset{TOF(X;x)}{\mapsto} \ket{X,Y,0,z \oplus xy},
\end{eqnarray*} 
which is indicated by the formulas on the output side. This, in turn, leads to the following corollary.

\begin{col}\label{cor:1}
An $n$-qubit Toffoli gate $TOF^n$ can be implemented with the cost not exceeding the sum of twice the cost of an $n$-qubit relative phase Toffoli gate $RTOF^n$ and the cost of the CNOT gate, using one ancilla qubit set to and returned in the value $\ket{0}$.  In other words, in the presence of such an ancilla,
$$Cost(TOF^n) \leq 2 \times Cost(RTOF^n) + Cost(CNOT).$$
\end{col}
This corollary may be reformulated for a different choice of the middle gate, {\em e.g.}, as follows: $Cost(TOF^n) \leq 2 \times Cost(RTOF^{n-1}) + Cost(TOF^3)$.

Other gate configurations are also supported by the relative phase Toffolis. The following Proposition complements the set of basic rules we base the proposed optimization approach on.

\begin{lemm}\label{lemma:main2}
Consider the conjugation of a controlled-$U$ gate $R_1CU(W;X,Y)$ implemented possibly up to some relative phase, by a pair of identical multiple control Toffoli gates, such as illustrated in (\ref{circ:main2}) on the left hand side.  Then, the following circuit identity holds for any unitary transformation $V$ over the qubit set $\{Z \cup a\}$ and any $S^YR_2TOF(Y,Z;a)$ (a type-$Y$ special form relative phase Toffoli gate):
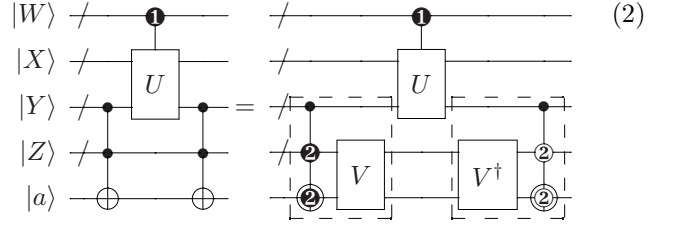
\begin{eqnarray}\label{circ:main2}
\Qcircuit @C=.5em @R=.83em @!R {
\lstick{\ket{W}} & {/} \qw 	& \qw		& \ctrlone{1}		& \qw	 	& \qw \\
\lstick{\ket{X}} & {/} \qw	& \qw		& \multigate{1}{U} 	& \qw		& \qw \\
\lstick{\ket{Y}} & {/} \qw	& \ctrl{1}	& \ghost{U} 		& \ctrl{1}	& \qw \\
\lstick{\ket{Z}} & {/} \qw	& \ctrl{1}	& \qw	 			& \ctrl{1}	& \qw \\
\lstick{\ket{a}} & \qw		& \targ		& \qw	 			& \targ		& \qw
}
&
\raisebox{-3.45em}{\hspace{1mm}=\hspace{2mm}}
&
\Qcircuit @C=0.5em @R=.7em @!R {
& {/} \qw 	& \qw			& \qw				& \ctrlone{1} 		 & \qw						& \qw	 		& \qw \\
& {/} \qw	& \qw			& \qw				& \multigate{1}{U} 	 & \qw						& \qw			& \qw \\
& {/} \qw	& \ctrl{1}		& \qw				& \ghost{U} 		 & \qw						& \ctrl{1}		& \qw \\
& {/} \qw	& \ctrltwo{1}	& \multigate{1}{V} 	& \qw	 			 & \multigate{1}{V^\dagger}	& \ctrlotwo{1}	& \qw \\
& \qw		& \targtwo		& \ghost{V}			& \qw	 			 & \ghost{V^\dagger}		& \targotwo		& \qw
\gategroup{3}{3}{5}{4}{0.5em}{--} \gategroup{3}{6}{5}{7}{0.5em}{--}
}
\end{eqnarray}
\end{lemm}
\begin{proof}
This proposition may be proved similarly to Proposition \ref{lemma:main1},
\begin{eqnarray*}
TOF(Y,Z;a)R_1CU(W;X,Y)TOF(Y,Z;a) \\
= TOF(Y,Z;a)D_2(Z;a)D_2^{-1}(Z;a) \\ R_1CU(W;X,Y)TOF(Y,Z;a) \\
= TOF(Y,Z;a)D_2(Z;a)R_1CU(W;X,Y) \\ D_2^{-1}(Z;a)TOF(Y,Z;a) \\
= S^YR_2TOF(Y,Z;a)R_1CU(W;X,Y) \\ S^YR_2^{-1}TOF(Y,Z;a) \\
= S^YR_2TOF(Y,Z;a)V(Z;a)V^{-1}(Z;a) \\ R_1CU(W;X,Y)S^YR_2^{-1}TOF(Y,Z;a) \\
= \big[ S^YR_2TOF(Y,Z;a)V(Z;a) \big] R_1CU(W;X,Y) \\ \big[ V^{-1}(Z;a) S^YR_2^{-1}TOF(Y,Z;a) \big]
\end{eqnarray*}

An alternate proof may be constructed via restricting $W, X, Y,$ and $Z$ to contain at most a single qubit each, and multiplying the corresponding matrices \cite{www:nb}.  The benefit of considering such a matrix multiplication is in the ability to show that the $S^YR_2TOF(Y,Z;a)$ turns out to be the relative phase Toffoli gate that allows most freedom in selecting relative phases for a general unitary $U$, allowing to formulate this proposition as an ``if-and-only-if'' statement.  Furthermore, looking at the matrices helps to expand the set of possible allowed relative phase replacements once $U$ is known. 
\end{proof}

The results of Propositions \ref{lemma:main1} and \ref{lemma:main2} may be generalized via introducing a control set $P$ that controls all three gates on the left hand side and well as all five gates on the right hand side, and a control set $Q$ that controls all gates except $V$.

Observe, that between the two Propositions they cover all situations when a relative phase controlled-$U$ is conjugated by a pair of multiple control Toffoli gates such that the targets of those multiple control Toffoli gates do not intersect with the $U$, resulting in the ability to replace a pair of multiple control Toffoli gates with a pair of simpler gates.  A similar circuit identity may be developed for the scenario when the target of the multiple control Toffolis intersects with the qubits used by the unitary $U$.  This circuit identity relies on the special form relative phase Toffoli gates.  We have not yet found practical examples where such circuit identity would yield an advantage and the results of Propositions \ref{lemma:main1} and \ref{lemma:main2} do not apply, but formulate the statement of the respective Proposition for completeness.

\begin{lemm}\label{lemma:main3}
The conjugation of the controlled unitary $U$ implemented up to a relative phase, $R_1CU(X;Y,Z,a)$, by a pair of the multiple control Toffoli gates $TOF(W,Z;a)$ allows the replacement of these multiple control Toffoli gates with the type-$\{Z \cup a\}$ special form relative phase version (up to a multiplication by any desired unitary $V(W)$) and its inverse, as follows:
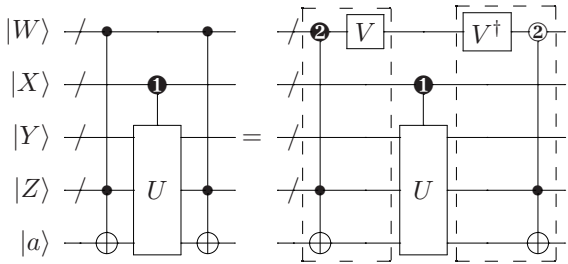
\begin{eqnarray*}\label{circ:main3}
\Qcircuit @C=.6em @R=1.1em @!R {
\lstick{\ket{W}} & {/} \qw 	& \ctrl{3}	& \qw	 				& \ctrl{3}	& \qw \\
\lstick{\ket{X}} & {/} \qw	& \qw		& \ctrlone{1} 			& \qw		& \qw \\
\lstick{\ket{Y}} & {/} \qw	& \qw		& \multigate{2}{U} 		& \qw		& \qw \\
\lstick{\ket{Z}} & {/} \qw	& \ctrl{1}	& \ghost{U} 			& \ctrl{1}	& \qw \\
\lstick{\ket{a}} & \qw		& \targ		& \ghost{U} 			& \targ 	& \qw
}
&
\raisebox{-4.0em}{\hspace{1mm}=\hspace{2mm}}
&
\Qcircuit @C=0.6em @R=.55em @!R {
& {/} \qw 	& \ctrltwo{3}	& \gate{V} 	& \qw	 			& \gate{V^\dagger}	& \ctrlotwo{3}	& \qw \\
& {/} \qw	& \qw			& \qw		& \ctrlone{1} 		& \qw 				& \qw			& \qw \\
& {/} \qw	& \qw			& \qw		& \multigate{2}{U} 	& \qw 				& \qw			& \qw \\
& {/} \qw	& \ctrl{1}		& \qw		& \ghost{U}		 	& \qw 				& \ctrl{1}		& \qw \\
& \qw		& \targ			& \qw		& \ghost{U} 		& \qw 				& \targ 		& \qw 
\gategroup{1}{3}{5}{4}{0.5em}{--} \gategroup{1}{6}{5}{7}{0.5em}{--}
}
\end{eqnarray*}
\end{lemm}
We do not include an explicit proof, but mention that it may be obtained similarly to that of Propositions \ref{lemma:main1} and \ref{lemma:main2}.  Furthermore, we note that the scenario where $R_1CU(X;Y,Z,a)$ is a diagonal gate, {\em e.g.}, a controlled-$R_z$ implemented up to a possible relative phase, is better handled by applying Proposition \ref{lemma:main1} than Proposition \ref{lemma:main3} ({\em e.g.}, see item \ref{ex:lemma1}, Subsection \ref{subsec:app}).  Indeed, Proposition \ref{lemma:main1} uses the most generic unspecified type relative phase Toffoli, and any controlled-$R_z$ may be thought of as a targetless gate ($|Z|=0$ in the statement of Proposition \ref{lemma:main1}) or otherwise, one may introduce a new target qubit that applies a global phase \cite[Figure 4.5]{bk:nc}.

\subsection{Applications}\label{subsec:app}

The principal circuit equalities (\ref{circ:main1}) and (\ref{circ:main2}) suggest a circuit optimization procedure by which a suitable pair of the multiple control Toffoli gates can be replaced with their relative phase or special form relative phase implementations up to the right hand multiplication by any desired unitary over the proper qubit set.  The rules may be used interchangeably and combined.  In particular, we next illustrate how the above approach can be applied to optimize the most popular constructs used to implement/decompose the multiple control Toffoli gates into simpler gates.  In the following discussions, we will omit unitaries $V$, with the understanding that if needs be, they may be added back in. 

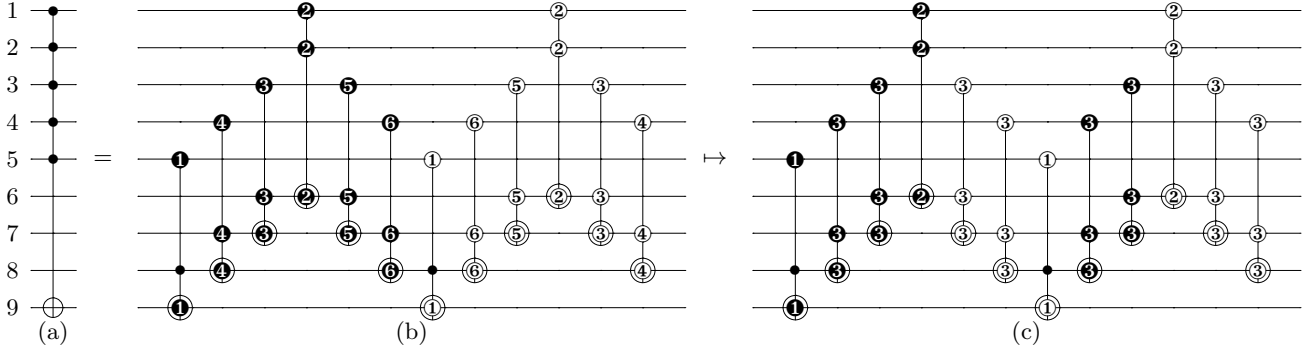
\begin{figure*}[t]
\centerline{
\begin{tabular}{cccccc}
\Qcircuit @C=0.5em @R=0.7em @!R {
\lstick{1} & \ctrl{1}	& \qw \\
\lstick{2} & \ctrl{1}	& \qw \\
\lstick{3} & \ctrl{1}	& \qw \\
\lstick{4} & \ctrl{1}	& \qw \\
\lstick{5} & \ctrl{4}	& \qw \\
\lstick{6} & \qw		& \qw \\
\lstick{7} & \qw		& \qw \\
\lstick{8} & \qw		& \qw \\
\lstick{9} & \targ		& \qw
}
&
\raisebox{-6.3em}{\hspace{1mm}=\hspace{2mm}}
&
\Qcircuit @C=0.7em @R=.5em @! {
& \qw 			& \qw 			& \qw 			& \ctrltwo{1} 	& \qw 			& \qw 			& \qw  			& \qw 		  & \qw 		 & \ctrlotwo{1} & \qw 		& \qw 			& \qw \\
& \qw 			& \qw 			& \qw 			& \ctrltwo{4} 	& \qw 			& \qw 			& \qw  			& \qw 		  & \qw 		 & \ctrlotwo{4} & \qw 		& \qw 			& \qw \\
& \qw 			& \qw 			& \ctrlthree{3} & \qw 			& \ctrlfive{3}	& \qw 			& \qw  			& \qw 		  & \ctrlofive{3}& \qw 		& \ctrlothree{3}& \qw 			& \qw \\
& \qw 			& \ctrlfour{3} 	& \qw 			& \qw 			& \qw 			& \ctrlsix{3}   & \qw  			& \ctrlosix{3}& \qw 		 & \qw 		& \qw 			& \ctrlofour{3} & \qw \\
& \ctrlone{3} 	& \qw 			& \qw 			& \qw 			& \qw 			& \qw 			& \ctrloone{3} 	& \qw 	  	  & \qw 		 & \qw 		& \qw 			& \qw 	 		& \qw \\
& \qw 			& \qw 			& \ctrlthree{1} & \targtwo 		& \ctrlfive{1}	& \qw 			& \qw  			& \qw 		  & \ctrlofive{1}& \targotwo& \ctrlothree{1}& \qw 			& \qw \\
& \qw 			& \ctrlfour{1} 	& \targthree	& \qw 			& \targfive		& \ctrlsix{1} 	& \qw  			& \ctrlosix{1}& \targofive	 & \qw 		& \targothree	& \ctrlofour{1} & \qw \\
& \ctrl{1} 		& \targfour		& \qw 			& \qw 			& \qw 			& \targsix		& \ctrl{1}	 	& \targosix	  & \qw 		 & \qw 		& \qw 			& \targofour	& \qw \\
& \targone 		& \qw 			& \qw 			& \qw 			& \qw 			& \qw 			& \targoone		& \qw 		  & \qw 		 & \qw 		& \qw 			& \qw 			& \qw 
} 
&
\raisebox{-6.3em}{\hspace{1mm}$\mapsto$\hspace{2mm}}
&
\Qcircuit @C=0.7em @R=.5em @! {
& \qw 			& \qw 			& \qw 			& \ctrltwo{1} 	& \qw 			& \qw 			& \qw  			& \qw 		  & \qw 		 & \ctrlotwo{1} & \qw 		& \qw 			& \qw \\
& \qw 			& \qw 			& \qw 			& \ctrltwo{4} 	& \qw 			& \qw 			& \qw  			& \qw 		  & \qw 		 & \ctrlotwo{4} & \qw 		& \qw 			& \qw \\
& \qw 			& \qw 			& \ctrlthree{3} & \qw 			& \ctrlothree{3}& \qw 			& \qw  			& \qw 		  & \ctrlthree{3}& \qw 		& \ctrlothree{3}& \qw 			& \qw \\
& \qw 			& \ctrlthree{3} & \qw 			& \qw 			& \qw 			& \ctrlothree{3}& \qw  			& \ctrlthree{3}& \qw 		 & \qw 		& \qw 			& \ctrlothree{3} & \qw \\
& \ctrlone{3} 	& \qw 			& \qw 			& \qw 			& \qw 			& \qw 			& \ctrloone{3} 	& \qw 	  	  & \qw 		 & \qw 		& \qw 			& \qw 	 		& \qw \\
& \qw 			& \qw 			& \ctrlthree{1} & \targtwo 		& \ctrlothree{1}& \qw 			& \qw  			& \qw 		  & \ctrlthree{1}& \targotwo& \ctrlothree{1}& \qw 			& \qw \\
& \qw 			& \ctrlthree{1} & \targthree	& \qw 			& \targothree	& \ctrlothree{1}& \qw  			& \ctrlthree{1}& \targthree	 & \qw 		& \targothree	& \ctrlothree{1} & \qw \\
& \ctrl{1} 		& \targthree	& \qw 			& \qw 			& \qw 			& \targothree	& \ctrl{1}	 	& \targthree  & \qw 		 & \qw 		& \qw 			& \targothree	& \qw \\
& \targone 		& \qw 			& \qw 			& \qw 			& \qw 			& \qw 			& \targoone		& \qw 		  & \qw 		 & \qw 		& \qw 			& \qw 			& \qw 
} \\
(a)&&(b)&&(c)
\end{tabular}
}
\caption{(a-b) Implementation of $TOF^6$ on the qubits $1-9$ using $S^{\{8\}}RTOF^3$ and its inverse, and $10$ $RTOF^3$ gates. (a-c) Implementation of $TOF^6$ using a narrow selection of the relative phase and special form relative phase Toffoli gates.}
\label{circ:tofn-3}
\end{figure*}

\begin{col} {\normalfont [Optimization of the construction reported in \cite[Lemma 7.2]{j:bbcd}.]}\label{cor2}
A multiple control Toffoli gate $TOF^n$ can be implemented by a circuit consisting of $4n-14$ relative phase Toffoli gates $RTOF^3$ and a type-$\{y\}$ special form relative phase Toffoli gate $S^yRTOF^3(x,y;z)$ and its inverse over a circuit with at least $2n-3$ qubits, such as illustrated in Figure \ref{circ:tofn-3}.
\end{col}

\begin{proof}
The numeric order of subscripts in the special form and relative phase Toffoli gates indicates the order in which the circuit equalities (\ref{circ:main2}) and (\ref{circ:main1}) are applied to the original circuit reported in \cite[Lemma 7.2]{j:bbcd} to obtain the desired simplified decomposition.  Observe that when during this process a pair of Toffoli gates $TOF^3(a,b;c)$ is replaced with a special form or a relative phase implementation, the circuit in the middle may be equivalent to a combination of a suitable multiple control Toffoli gate---possibly up to a relative phase, and a transformation on the qubits outside the set $\{a,b,c\}$.  This latter transformation may be factored out, thereby allowing all circuit alternations to retain the original functional correctness.

Finally, observe that the identities (\ref{circ:main1}) and (\ref{circ:main2}) may be used in a number of different ways, resulting in different constructions, and not just the particular one selected in the statement of the Corollary.  In Figure \ref{circ:tofn-3}(b) we used one of such constructions that minimizes the number of the special form relative phase Toffoli gates to gain most freedom in substituting Toffoli gates with their relative phase implementations. In Figure \ref{circ:tofn-3}(c) we furthermore restricted the number of potentially different $RTOF$ gates via making the following assignments: $R_4TOF:=R_3TOF$, $R_5TOF:=R_3^{-1}TOF$, and $R_6TOF:=R_3^{-1}TOF$.  This implementation will be used later in the paper. 
\end{proof}

\begin{col} {\normalfont [Optimization of the construction reported in \cite[Lemma 7.3]{j:bbcd}.]}
A multiple control Toffoli gate $TOF^n$ can be implemented by a circuit consisting of two relative phase Toffoli gates $RTOF^k$ and two special form relative phase Toffoli gates $SRTOF^{n-k+2}$ over a circuit with at least $n+1$ qubits, such as illustrated next:
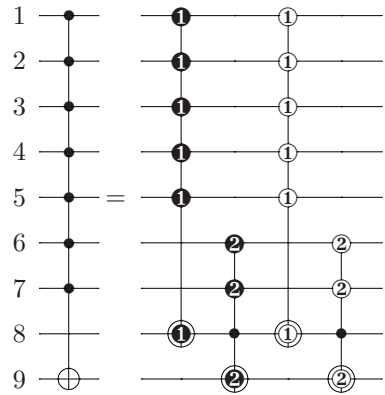
\begin{eqnarray*}\label{circ:tof1}
\Qcircuit @C=0.7em @R=0.9em @!R {
\lstick{1} & \ctrl{1}	& \qw \\
\lstick{2} & \ctrl{1}	& \qw \\
\lstick{3} & \ctrl{1}	& \qw \\
\lstick{4} & \ctrl{1}	& \qw \\
\lstick{5} & \ctrl{1}	& \qw \\
\lstick{6} & \ctrl{1}	& \qw \\
\lstick{7} & \ctrl{2}	& \qw \\
\lstick{8} & \qw		& \qw \\
\lstick{9} & \targ		& \qw
}
&
\raisebox{-6.9em}{\hspace{1mm}=\hspace{2mm}}
&
\Qcircuit @C=1em @R=.7em @!R {
& \ctrlone{1} 	& \qw			& \ctrloone{1}	& \qw			& \qw 	\\
& \ctrlone{1} 	& \qw			& \ctrloone{1}	& \qw			& \qw 	\\
& \ctrlone{1} 	& \qw			& \ctrloone{1}	& \qw			& \qw 	\\
& \ctrlone{1} 	& \qw			& \ctrloone{1}	& \qw			& \qw 	\\
& \ctrlone{3} 	& \qw			& \ctrloone{3}	& \qw			& \qw 	\\
& \qw			& \ctrltwo{1} 	& \qw			& \ctrlotwo{1}	& \qw	\\
& \qw			& \ctrltwo{1} 	& \qw			& \ctrlotwo{1}	& \qw	\\
& \targone		& \ctrl{1} 		& \targoone		& \ctrl{1}		& \qw	\\
& \qw			& \targtwo	 	& \qw			& \targotwo		& \qw
}
\end{eqnarray*}
\end{col}
\begin{proof}
To obtain this construction, both circuit identities (\ref{circ:main1}) and (\ref{circ:main2}) need to be applied once, in any order.
\end{proof}

\begin{col}{\normalfont [Optimization of the construction in \cite[page 184]{bk:nc}.]}\label{cor4}
A multiple control gate $C^nU$ can be implemented by a circuit consisting of $2n-2$ relative phase Toffoli gates $RTOF^3$ and one $CU$ gate over a circuit with at least $2n$ qubits of which some $n-1$ qubits are set to and returned in the value $\ket{0}$, such as illustrated next:
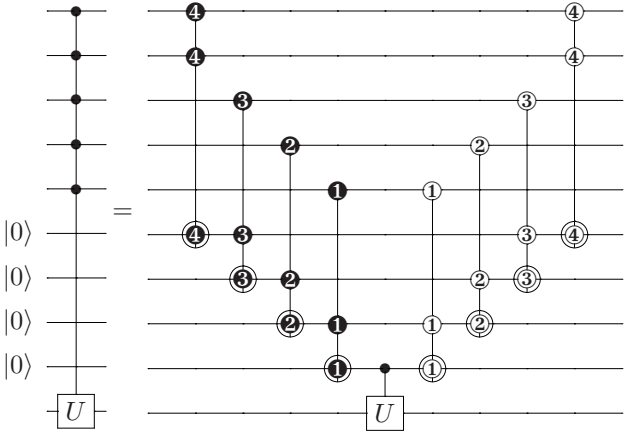
\begin{eqnarray*}\label{circ:tofnzeros}
\Qcircuit @C=0.4em @R=0.4em @!R {
				 & \ctrl{1}	& \qw \\
				 & \ctrl{1}	& \qw \\
				 & \ctrl{1}	& \qw \\
				 & \ctrl{1}	& \qw \\
				 & \ctrl{5}	& \qw \\
\lstick{\ket{0}} & \qw		& \qw \\
\lstick{\ket{0}} & \qw		& \qw \\
\lstick{\ket{0}} & \qw		& \qw \\
\lstick{\ket{0}} & \qw		& \qw \\
				 & \gate{U}	& \qw
}
&
\raisebox{-7.6em}{\hspace{1mm}=\hspace{2mm}}
&
\Qcircuit @C=0.4em @R=.3em @! {
& \ctrlfour{1} 	& \qw 			& \qw 			& \qw  			& \qw		& \qw  			& \qw 			& \qw 			& \ctrlofour{1} & \qw  \\
& \ctrlfour{4} 	& \qw 			& \qw 			& \qw  			& \qw		& \qw  			& \qw 			& \qw 			& \ctrlofour{4} & \qw  \\
& \qw 			& \ctrlthree{3} & \qw 			& \qw  			& \qw		& \qw  			& \qw 			& \ctrlothree{3}& \qw 			& \qw  \\
& \qw 			& \qw 			& \ctrltwo{3} 	& \qw  			& \qw		& \qw  			& \ctrlotwo{3} 	& \qw 			& \qw 			& \qw  \\
& \qw 			& \qw 			& \qw 			& \ctrlone{3} 	& \qw		& \ctrloone{3}	& \qw 			& \qw 			& \qw 			& \qw  \\
& \targfour 	& \ctrlthree{1} & \qw 			& \qw  			& \qw		& \qw  			& \qw 			& \ctrlothree{1}& \targofour	& \qw \\
& \qw 			& \targthree	& \ctrltwo{1} 	& \qw  			& \qw		& \qw  			& \ctrlotwo{1} 	& \targothree	& \qw 			& \qw  \\
& \qw 			& \qw 			& \targtwo		& \ctrlone{1} 	& \qw		& \ctrloone{1} 	& \targotwo		& \qw 			& \qw 			& \qw \\
& \qw 			& \qw 			& \qw 			& \targone 		& \ctrl{1}	& \targoone		& \qw 			& \qw 			& \qw 			& \qw  \\
& \qw 			& \qw 			& \qw 			& \qw 			& \gate{U}	& \qw 			& \qw 			& \qw 			& \qw 			& \qw  
}
\end{eqnarray*}
\end{col}
\begin{proof}
The circuit identity (\ref{circ:main1}) is applied $n-1$ times.
\end{proof}

The implementation in \cite[equation (13)]{j:s} optimizes the depth of the circuit \cite[page 184]{bk:nc}, but does not prevent our construction from being applied.  We formalize this observation in the following Corollary.
\begin{col}{\normalfont [Optimization/generalization of the construction in \cite[equation (13)]{j:s}.]}\label{cor5}
A multiple control gate $C^nU$ can be implemented by a circuit consisting of $2n-2$ relative phase Toffoli gates $RTOF^3$ and one $CU$ gate over a circuit with at least $2n$ qubits of which some $n-1$ qubits are set to and returned in the value $\ket{0}$, such as illustrated next:
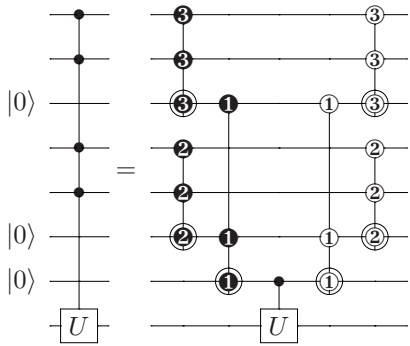
\begin{eqnarray*}\label{circ:c-uparallel}
\Qcircuit @C=0.4em @R=0.4em @!R {
				 & \ctrl{1}	& \qw \\
				 & \ctrl{2}	& \qw \\
\lstick{\ket{0}} & \qw		& \qw \\
				 & \ctrl{1}	& \qw \\
				 & \ctrl{3}	& \qw \\
\lstick{\ket{0}} & \qw		& \qw \\
\lstick{\ket{0}} & \qw		& \qw \\
				 & \gate{U}	& \qw 
}
&
\raisebox{-5.9em}{\hspace{1mm}=\hspace{2mm}}
&
\Qcircuit @C=0.7em @R=.4em @!R {
& \ctrlthree{1}	& \qw 			& \qw  		& \qw			& \ctrlothree{1}& \qw  \\
& \ctrlthree{1}	& \qw 			& \qw  		& \qw			& \ctrlothree{1}& \qw  \\
& \targthree	& \ctrlone{3}	& \qw  		& \ctrloone{3}	& \targothree	& \qw  \\
& \ctrltwo{1}	& \qw	 		& \qw  		& \qw			& \ctrlotwo{1}	& \qw  \\
& \ctrltwo{1}	& \qw 			& \qw	 	& \qw			& \ctrlotwo{1}	& \qw  \\
& \targtwo	 	& \ctrlone{1}	& \qw  		& \ctrloone{1}	& \targotwo		& \qw  \\
& \qw	 		& \targone 		& \ctrl{1}	& \targoone		& \qw  			& \qw  \\
& \qw 			& \qw			& \gate{U} 	& \qw			& \qw	 		& \qw  
}
\end{eqnarray*}
\end{col}

Some other optimizations include the following.
\begin{enumerate}
\item Circuit in \cite[Lemma 7.5]{j:bbcd} may rely on the simpler relative phase multiple control Toffoli gate and its inverse, rather than two multiple control Toffoli gates (gates \#2 and \#4 on the right hand side). 
\item Circuit in \cite[Lemma 7.9]{j:bbcd} may rely on the simpler special form relative phase multiple control Toffoli gate and its inverse, rather than two multiple control Toffoli gates (gates \#2 and \#4 on the right hand side).
\item Circuit in \cite[Lemma 7.11]{j:bbcd} may rely on the simpler relative phase multiple control Toffoli gate and its inverse, rather than two multiple control Toffoli gates (gates \#1 and \#3 on the right hand side).
\item\label{ex:lemma1} Circuit in \cite[Figure 3]{c:ccdf} may rely on the simpler relative phase Toffoli gates and their inverses, as is best seen via applying Proposition \ref{lemma:main1}.
\end{enumerate}

\section{Optimizing implementations of the multiple control Toffoli gates using the existing relative phase Toffoli circuits}

In this section we study in detail how to optimize the implementations of the multiple control Toffoli gates, show that all of the known optimized implementations can be explained by the means of the relative phase Toffoli substitutions described in this work, and report some new optimized circuits.

\subsection{Circuit cost}

The question of the efficiency of implementing a certain transformation requires one to formally define a circuit cost.  Depending on the definition of cost, certain circuits will be preferred over  other circuits.  

There are a number of different definitions of the circuit cost used in the literature, each originating from considering certain specific requirements.  At the highest abstraction level, firstly, one needs to determine if they are dealing with logical level or physical level circuits.  

In the former case, one has to derive the protocols and compute the costs of the constructible fault-tolerant gates, given the selected approach to error correction.  Within this framework Clifford$+T$ circuits received a significant attention.  This is because Clifford gates such as Pauli-$X$, $Y$, $Z$, Hadamard, Phase, and CNOT are believed to be relatively inexpensive to implement fault tolerantly on the logical level.  The non-Clifford gate $T$, or any other constructible non-Clifford gate required for computational universality, is more difficult to generate.  The known approaches employ state purification and gate teleportation as a means of generating the $T$ gate, that can get quite costly in the realistic systems \cite{ar:bk}.  As a result, the cost of the implementation of a logical circuit can be very crudely approximated by the number of the $T$ gates used. 

In the case of physical level circuits, one is limited to the ability of the controlling apparatus to apply transformations to the physical quantum information processing system of choice.  There is a great variety of the possibilities here.  We consider a simple and popular weak interaction model, where the single-qubit gates can be implemented efficiently, and of the two-qubit gates, that take considerably more effort to implement, we have just the CNOT gate.  The cost of the circuits can thus be evaluated via counting the number of the CNOT gates in the single-qubit and CNOT gate circuits.  Despite apparent oversimplification, there is a specific promising quantum information processing approach, where exactly this formula describes the circuit cost at a high abstraction level.  Indeed, trapped ions with Molmer-Sorenson gate \cite{ar:somo} operate in the weak coupling regime (two-qubit gates take roughly $10-20$ fold effort to implement compared to arbitrary single-qubit gates), and Molmer-Sorenson gate itself is equivalent to the CNOT up to a conjugation by a pair of $R_Z(a)$ and $R_Z(-a)$ gates on both qubits, for a proper choice of parameter $a$, and a few single-qubit Phase and Hadamard gates. 

An advantage of measuring the cost of the circuit implementations by the $T$-count and the CNOT-count is due to the popularity of these circuit cost metrics in the literature, and the ability to compare relative phase inspired implementations developed in this work to the known ones. 

Disadvantages of using either one of these two circuit cost metrics are numerous.  Neither circuit metric accounts for:
\begin{itemize}
\item the depth, that could be more important than the gate count, especially when one is, quite naturally, concerned with the speed of the computation given by a quantum circuit rather than just its size; 
\item the connectivity pattern of the qubits. Indeed, physical space spans only three dimensions, and every qubit cannot be connected to every other qubit in a scalable fashion within a finite-dimensional space; or 
\item the number of ancillary qubits used, that is particularly important on the physical level.  The number of ancillary qubits used also influences the efficiency of connections between primary qubits.  This is because both primary qubits and ancillary qubits share same physical space and yet need to be as close to each other as possible for higher efficiency.
\end{itemize} 
These are all very important practical considerations.  However, our goal is to demonstrate the advantages of the framework introduced in this paper for designing efficient circuits, therefore we restrict the attention to the above two simplistic metrics.  We furthermore encourage to apply the techniques from this paper to designing efficient circuits in the scenario where the details of the circuit cost function are known. 

\subsection{Toffoli and Toffoli-4 gates up to a relative phase}
Firstly, recall a circuit implementing the Toffoli gate $TOF(a,b;c)$ itself: 
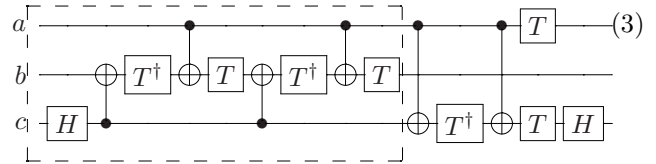
\begin{eqnarray}\label{circ:tof}
\Qcircuit @C=0.3em @R=.4em @!R {
\lstick{a} & \qw 		& \qw		& \qw				& \ctrl{1}	& \qw		& \qw 		& \qw 				& \ctrl{1} 	& \qw 		& \ctrl{2} 	& \qw 			 & \ctrl{2}	& \gate{T} 	& \qw		& \qw \\
\lstick{b} & \qw 		& \targ		& \gate{T^\dagger}	& \targ		& \gate{T}	& \targ		& \gate{T^\dagger} 	& \targ 	& \gate{T} 	& \qw	 	& \qw			 & \qw 		& \qw 		& \qw		& \qw \\
\lstick{c} & \gate{H}	& \ctrl{-1}	& \qw			 	& \qw		& \qw		& \ctrl{-1}	& \qw			 	& \qw		& \qw		& \targ		& \gate{T^\dagger}& \targ 	& \gate{T} 	& \gate{H} 	& \qw 
\gategroup{1}{2}{3}{10}{1.5em}{--} 
}
\end{eqnarray}
This circuit may be drawn in many different ways using no more than the minimal numbers of $6$ CNOT gates and $7$ $T/T^\dagger$ gates, however, we prefer this form since it has the largest number of gates operating on the qubits $a$ and $c$ after no more gates are being applied to the qubit $b$.

\begin{figure}
\centerline{
\Qcircuit @C=0.7em @R=.4em @!R {
\lstick{a} & \ctrl{2} 	& \qw 		& \qw		& \qw		& \qw 				& \ctrl{2}	& \qw 		& \qw 		& \qw 				& \qw 		& \qw \\
\lstick{b} & \qw		& \qw	 	& \qw		& \ctrl{1}	& \qw 				& \qw		& \qw 		& \ctrl{1} 	& \qw 				& \qw 		& \qw \\
\lstick{c} & \gate{Z} 	& \gate{H} 	& \gate{T} 	& \targ 	& \gate{T^\dagger}	& \targ		& \gate{T} 	& \targ 	& \gate{T^\dagger} 	& \gate{H} 	& \qw 
\gategroup{1}{3}{3}{11}{1em}{--}
}
}
\caption{Toffoli gate implemented up to a relative phase: gates 1-10 implement a type-$\{c\}$ special relative phase Toffoli gate, known as the controlled-controlled-$iX$ in \cite{j:s}, whereas circuit with gates 2-10 implements some generic relative phase Toffoli gate.  The controlled-$Z$ gate $CZ(a;c)$ may commute through the Hadamard $H(c)$, at which point it will change into $CNOT(a;c)$, and the circuit will show in an alternate form.  It may be established, via applying the result of Corollary \ref{cor:1}, that the CNOT count of the circuit with gates 2-10 is optimal.} \label{fig:rt}
\end{figure}
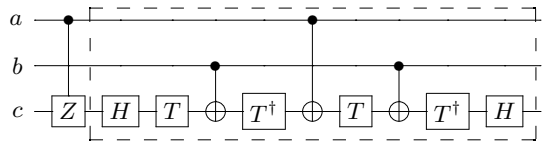 

Literature encounters two apparently related implementations of the Toffoli gate up to a relative phase \cite[page 183]{bk:nc} and \cite{j:s}, that we summarize in one distilled picture, see Figure \ref{fig:rt}.  There are more symmetries and properties to this circuit than those that necessarily meet the eye on the first glance.  In particular, 
\begin{itemize}
\item Gates 1-10 implement a type-$\{c\}$ special form relative phase Toffoli gate $S^cRTOF(a,b;c)=$ $diag\left\{1,1,1,1,1,1,\left(\begin{array}{cc} 0 & i \\ i & 0 \end{array} \right)\right\}$, whereas gates 2-10 implement a relative phase Toffoli gate $RTOF(a,b;c)=diag\left\{1,1,1,1,1,-1,\left(\begin{array}{cc} 0 & -i \\ i & 0 \end{array} \right)\right\}$. 
\item First gate, the controlled-$Z$, can be moved to the end of the circuit, resulting in the construction of the type-$\{c\}$ special form relative phase Toffoli gate $S^cRTOF(a,b;c)=$ $diag\left\{1,1,1,1,1,1,\left(\begin{array}{cc} 0 & -i \\ -i & 0 \end{array} \right)\right\}$.
\item Simultaneous substitution $T \mapsto T^\dagger$ and $T^\dagger \mapsto T$ allows constructing more circuits implementing a relative phase Toffoli gate.
\item The circuit given by the gates 2-10 is self-inverse. 
\item Qubits $a$ and $b$ may be interchanged. Applying this operation gives modified relative phase Toffoli circuits.
\item Adding gates $T^m(a)$ and $T^n(b)$ (powers of the $T$ gate), where $m,n \in \{0,1,...,7\}$, to both the beginning and the end of the circuit in Figure \ref{fig:rt} allows constructing more relative phase Toffoli gates.
\item Consider gates 3-9.  Using the CNOT-$T$ algebra terminology \cite{ar:amm,j:s}, the $T$ gate is being applied to $\{c,-(b \oplus c), a \oplus b \oplus c, -(a \oplus c)\}$ (negative sign indicates the application of $T^\dagger$). Instead, we may apply the $T$ gate to $\{c,b \oplus c, -(a \oplus b \oplus c), -(a \oplus c)\}$. Then, the circuit we obtain looks as follows:
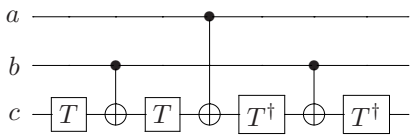
\[
\Qcircuit @C=0.7em @R=.4em @!R {
\lstick{a} & \qw		& \qw		& \qw 				& \ctrl{2}	& \qw 				& \qw 		& \qw 				& \qw \\
\lstick{b} & \qw		& \ctrl{1}	& \qw 				& \qw		& \qw 				& \ctrl{1} 	& \qw 				& \qw \\
\lstick{c} & \gate{T} 	& \targ 	& \gate{T}			& \targ		& \gate{T^\dagger} 	& \targ 	& \gate{T^\dagger} 	& \qw 
}
\]
Observe how similar it is to \cite[page 183]{bk:nc}---essentially, $Y$ rotations are replaced by $Z$ rotations.  Optimality of the above circuit employing $R_Y$ rotations in place of $T$ (sometimes known as Margolus gate) was shown in \cite{ar:sk}.  Conjugating this circuit by a pair of Hadamard gates on the qubit $c$ allows to obtain a relative phase Toffoli $RTOF(a,\bar{b};c)$, where $\bar{b}$ denotes the negative control.  Similarly, if the $T/T^\dagger$ gates of the circuit in Figure \ref{fig:rt}, gates 3-9, were replaced with $R_Y(\pi/4)/R_Y(-\pi/4)$, as illustrated next, 
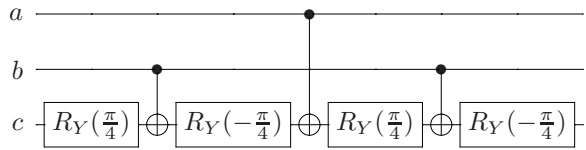
\[
\Qcircuit @C=0.3em @R=.4em @!R {
\lstick{a} & \qw				& \qw		& \qw 					& \ctrl{2}	& \qw 				& \qw 		& \qw 					& \qw \\
\lstick{b} & \qw				& \ctrl{1}	& \qw 					& \qw		& \qw 				& \ctrl{1} 	& \qw 					& \qw \\
\lstick{c} & \gate{R_Y(\frac{\pi}{4})} 	& \targ 	& \gate{R_Y(-\frac{\pi}{4})}	& \targ		& \gate{R_Y(\frac{\pi}{4})} & \targ 	& \gate{R_Y(-\frac{\pi}{4})}	& \qw 
}
\]
we would have obtained an $RTOF(a,\bar{b};c)$. 
\end{itemize} 

We found no relative phase Toffoli-4 implementations in the literature, but realized that one may be constructed as follows. Consider circuit in Figure \ref{fig:rt}, gates 2-10.  Replace $CNOT(a;c)$ with a type-$\{c\}$ special form relative phase Toffoli gate $S^{c}RTOF(x,a;c)$; this operation introduces a new qubit, $x$.  The result is an $RTOF(x,a,b;c)$.  Figure \ref{circ:rtof4} illustrates the result of such a procedure for $SRTOF$ selection per Figure \ref{fig:rt} (observe that the controlled-$Z$ was commuted through the Hadamard gate to obtain the CNOT).  In the matrix form, the gate looks as follows, $diag\left\{1,1,1,1,1,1,1,1,1,1,1,1,i,-i,\left(\begin{array}{cc} 0 & 1 \\ -1 & 0 \end{array} \right)\right\}$.

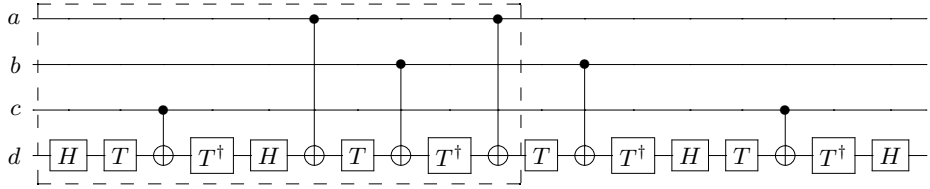
\begin{figure*}[t]
\centerline{
\Qcircuit @C=0.7em @R=.4em @!R {
\lstick{a} & \qw 		& \qw		& \qw		& \qw 				& \qw 		& \ctrl{3} 	& \qw		& \qw		& \qw 				& \ctrl{3}	& \qw 		& \qw 		& \qw 				& \qw 		& \qw 		& \qw 		& \qw 				& \qw 		& \qw \\
\lstick{b} & \qw 		& \qw		& \qw		& \qw 				& \qw 		& \qw 		& \qw		& \ctrl{2}	& \qw 				& \qw		& \qw 		& \ctrl{2}	& \qw 				& \qw 		& \qw 		& \qw 		& \qw 				& \qw 		& \qw \\
\lstick{c} & \qw	 	& \qw		& \ctrl{1}	& \qw 				& \qw	 	& \qw 		& \qw		& \qw		& \qw 				& \qw		& \qw 		& \qw	 	& \qw 				& \qw 		& \qw 		& \ctrl{1} 	& \qw 				& \qw 		& \qw \\
\lstick{d} & \gate{H} 	& \gate{T} 	& \targ 	& \gate{T^\dagger}	& \gate{H} 	& \targ 	& \gate{T} 	& \targ 	& \gate{T^\dagger}	& \targ		& \gate{T} 	& \targ 	& \gate{T^\dagger} 	& \gate{H} 	& \gate{T} 	& \targ 	& \gate{T^\dagger} 	& \gate{H} 	& \qw 
\gategroup{1}{2}{4}{11}{1em}{--}
}
}
\caption{Circuit implementing Toffoli-4 up to a relative phase, $RTOF(a,b,c;d)$.}\label{circ:rtof4}
\end{figure*}

\subsection{Results of the simplification}
Since T-count optimal and CNOT-count optimal implementations of the three-qubit Toffoli gate are known, we will concentrate on the Toffoli-4 and larger gates.  This section is not meant to report complete results of the optimization that is possible to obtain (indeed, there is no guarantee there are no better relative phase Toffoli-4 gates to be used, and we did not look for the relative phase Toffoli-5 and larger gates), rather show a clear advantage of using relative phase and special form relative phase Toffoli gates and motivate their further in-depth study. 

Consider Toffoli-4 implementation via a circuit with Clifford+$T$ gates.  Using matrix determinant argument, one may establish that the Toffoli-4 may not be implemented unless at least one ancilla qubit is available.  This is because the determinant of the $16 \times 16$ matrix representing the Toffoli-4 evaluates to the number $(-1)$, whereas the determinants of all Clifford+$T$ library gates, when viewed as $16 \times 16$ matrices, are equal to $1$.  By composing the products of matrices with determinant $1$ it is impossible to obtain a matrix with determinant $(-1)$.  As a result, at least one ancilla is required. 

Once we have established that an ancilla qubit is required, there are two options for the kind of ancilla qubit it is.  One, more restrictive, prescribes that the ancilla be available in the state $\ket{0}$; the other provides the ancilla in some unknown state, $\ket{x}$.  In both cases, when implementing Toffoli-4 with the help of an ancilla, special care needs to be taken to return the value of ancilla to its original state.  We consider both cases next.

\vspace{1mm}\noindent {\bf Optimization of Toffoli-4.}
\begin{itemize}
\item {\it Ancilla $\ket{0}$, minimizing $T$ count.} Literature encounters two results, \cite{ar:amm} and \cite{j:s}, both based on the optimization of \cite[page 184]{bk:nc}.  In particular, \cite{ar:amm} reports an optimized circuit with $15$ $T$ gates (down from unoptimized $21$), and \cite{j:s} observes that two Toffolis can be replaced with the relative phase Toffoli called the controlled-controlled-$iX$, Figure \ref{fig:rt}, which explains the optimization obtained in \cite{ar:amm}. Our solution uses a somewhat simpler relative phase Toffoli, see Figure \ref{fig:rt}, dashed (gates 2-10), to obtain $TOF^4(a,b,c;d)$: 
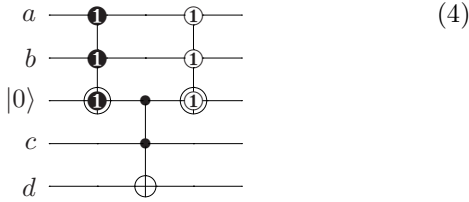
\begin{eqnarray}\label{circ:tof4zero}
\Qcircuit @C=0.8em @R=.6em @! {
\lstick{a} 			& \ctrlone{1} 	& \qw 		& \ctrloone{1}	& \qw	\\
\lstick{b} 			& \ctrlone{1} 	& \qw 		& \ctrloone{1}	& \qw	\\
\lstick{\ket{0}} 	& \targone 		& \ctrl{1} 	& \targoone		& \qw	\\
\lstick{c} 			& \qw			& \ctrl{1} 	& \qw	 		& \qw	\\
\lstick{d} 			& \qw			& \targ		& \qw	  		& \qw	
}
\end{eqnarray}
There is no advantage in the number of $T$ gates.  However, our solution explains both known circuits and features a smaller overall gate count.
\item {\it Ancilla $\ket{0}$, minimizing CNOT count.}  \cite{j:s} uses controlled-controlled-$iX$ to obtain an implementation with $14$ CNOTs.  To our knowledge, this was the best known result in the literature to date.  Our construction, (\ref{circ:tof4zero}), requires only $12=3+6+3$ CNOT gates, since our relative phase Toffoli (Figure \ref{fig:rt}, dashed) requires one less CNOT gate.  Observe, that per \cite{ar:sm} the lower bound for the number of CNOT gates is $8$.  Therefore, our 12-CNOT construction may not be improved by more than $4$ CNOT gates. 
\item {\it Arbitrary single-qubit ancilla, minimizing $T$ count.} The best known solution, \cite{ar:amm}, optimizes the $28$ $T$ gate implementation from \cite[Lemma 7.2]{j:bbcd}.  The result is a circuit with 16 $T$ gates. Our solution matches this solution, and in fact explains how it works.  Indeed, we obtain the desired $TOF^4(a,b,c;d)$ as follows:
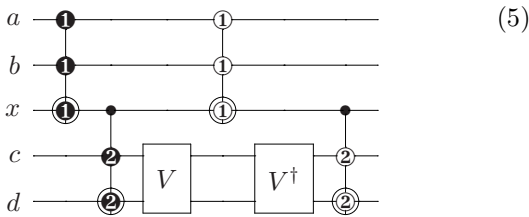
\begin{eqnarray}\label{circ:tof4any}
\Qcircuit @C=0.7em @R=.7em @!R {
\lstick{a} 			& \ctrlone{1} 	& \qw 			& \qw 				& \ctrloone{1}	& \qw 				 		& \qw 			& \qw	\\
\lstick{b} 			& \ctrlone{1} 	& \qw 			& \qw 				& \ctrloone{1}	& \qw 				 		& \qw			& \qw	\\
\lstick{x}		 	& \targone		& \ctrl{1}	 	& \qw		 		& \targoone		& \qw 				 		& \ctrl{1} 		& \qw	\\
\lstick{c} 			& \qw 			& \ctrltwo{1} 	& \multigate{1}{V} 	& \qw		 	& \multigate{1}{V^\dagger} 	& \ctrlotwo{1} 	& \qw	\\
\lstick{d} 			& \qw			& \targtwo		& \ghost{V}		 	& \qw		  	& \ghost{V^\dagger}			& \targotwo    	& \qw	
}
\end{eqnarray}
where $x$ is the ancilla qubit in an unknown state, $R_1TOF(a,b;x)$ is the relative phase Toffoli per Figure \ref{fig:rt}, dashed; and $S^xR_2TOF(x,c;d)V(c,d)$ pair is given by (\ref{circ:tof})--dashed.  Essentially, $V(c,d)$ is designed such as to undo all gates applied to the qubits $c$ and $d$ at the end of the implementation given by (\ref{circ:tof}).  We have not found a suitable special relative phase Toffoli gate implementation that is different from the implementation of the Toffoli gate itself, per (\ref{circ:tof}), and giving a better optimization once combined with proper $V(c,d)$.  The resulting $T$-count of our construction is thus $16=4+(7-3)+4+(7-3)$.  Apart from the matching number of $T$ gates, our solution contains fewer Clifford gates ({\em e.g.}, $14$ CNOTs vs $54$ CNOTs in \cite{ar:amm}), and may also be rewritten as a $T$-depth $4$ circuit ($T$-depth 1 per each relative phase Toffoli stage) at the cost of a higher number of ancillae and a higher number of CNOT gates. 
\item {\it Arbitrary single-qubit ancilla, minimizing CNOT count.} Using CNOT-optimal implementation of the controlled-controlled-$iX$ from \cite{j:s} over \cite[Lemma 7.2]{j:bbcd} would yield a circuit with $20$ CNOT gates, as is done in \cite{www:q}.  The original circuit, \cite[Lemma 7.2]{j:bbcd}, uses 24 CNOT gates after each Toffoli is substituted with their CNOT-optimal implementation.  Our construction, (\ref{circ:tof4any}), contains $14=3+4+3+4$ CNOT gates.
\end{itemize} 
Observe that the above implementations, if considered as circuits over Clifford+$T$ library, use the minimal number of ancillae, being one. 

\vspace{1mm}\noindent {\bf Optimization of Toffoli-5.}

One may once again apply the determinant argument to establish that the Toffoli-5 gate needs at least one ancilla to be available before it may be implemented as a Clifford+$T$ circuit.
 
\begin{itemize}
\item {\it All ancillae in the state $\ket{0}$, minimizing $T$ count.}  The best known solution is given by \cite{ar:amm} via an optimization of the construction in \cite[page 184]{bk:nc}, and explained by \cite{j:s} to be a four controlled-controlled-$iX$ and one Toffoli circuit.  The $T$-count is $23$ and both known solutions use two ancillae.  Our solution implementing $TOF^5(a,b,c,d;e)$ is as follows:  
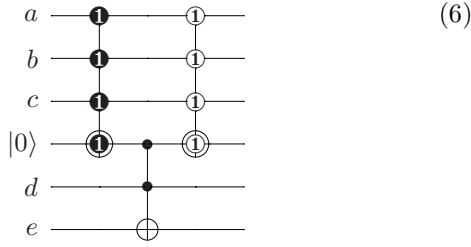
\begin{eqnarray}\label{circ:tof5zero}
\Qcircuit @C=0.8em @R=.6em @! {
\lstick{a} 			& \ctrlone{1} 	& \qw 		& \ctrloone{1}	& \qw	\\
\lstick{b} 			& \ctrlone{1} 	& \qw 		& \ctrloone{1}	& \qw	\\
\lstick{c} 			& \ctrlone{1} 	& \qw 		& \ctrloone{1}	& \qw	\\
\lstick{\ket{0}} 	& \targone 		& \ctrl{1} 	& \targoone		& \qw	\\
\lstick{d} 			& \qw			& \ctrl{1} 	& \qw		 	& \qw	\\
\lstick{e} 			& \qw			& \targ		& \qw		  	& \qw	
}
\end{eqnarray}
per $R_1TOF^4$ implementation found in Figure \ref{circ:rtof4} and Toffoli implementation from (\ref{circ:tof}).  Our solution uses $23=8+7+8$ $T$ gates, relies on only one ancilla, and has a smaller total number of gates compared to the previously known constructions. 
\item {\it All ancillae in the state $\ket{0}$, minimizing CNOT count.} The construction from \cite{j:s} gives the best known CNOT count of $22$ over a circuit that uses two ancillae.  Our circuit (\ref{circ:tof5zero}) contains $18$ CNOTs and uses only one ancilla.  Recall that the lower bound for the number of CNOT gates is $10$ \cite{ar:sm}.
\item {\it All ancillae in an unknown state, minimizing $T$ count.}  The best known solution is given in \cite{ar:amm} and features $28$ $T$ gates. Our solution implementing $TOF^5(a,b,c,d;e)$ is as follows:  
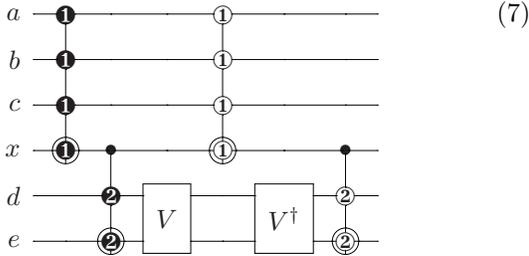
\begin{eqnarray}\label{circ:tof5any}
\Qcircuit @C=0.7em @R=.7em @!R {
\lstick{a} 			& \ctrlone{1} 	& \qw 			& \qw 				& \ctrloone{1}	& \qw 				 		& \qw 			& \qw	\\
\lstick{b} 			& \ctrlone{1} 	& \qw 			& \qw 				& \ctrloone{1}	& \qw 				 		& \qw 			& \qw	\\
\lstick{c} 			& \ctrlone{1} 	& \qw 			& \qw 				& \ctrloone{1}	& \qw 				 		& \qw			& \qw	\\
\lstick{x}		 	& \targone		& \ctrl{1}	 	& \qw		 		& \targoone 	& \qw 				 		& \ctrl{1}	 	& \qw	\\
\lstick{d} 			& \qw 			& \ctrltwo{1} 	& \multigate{1}{V} 	& \qw	 		& \multigate{1}{V^\dagger} 	& \ctrlotwo{1} 	& \qw	\\
\lstick{e} 			& \qw			& \targtwo		& \ghost{V}		 	& \qw	  		& \ghost{V^\dagger}			& \targotwo   	& \qw	
}
\end{eqnarray}
where $x$ is the ancilla qubit in an unknown state, $R_1TOF^4$ is the relative phase Toffoli from Figure \ref{circ:rtof4}, and $S^xR_2TOF(x,d;e)V(d,e)$ pair is given by (\ref{circ:tof})--dashed.  Observe, that the overall number of $T$ gates is $24=8+(7-3)+8+(7-3)$, we use one less ancilla compared to the best known construction, and a smaller overall number of the non-$T$ gates. 

We can furthermore explain how to obtain the solution with $12k-20$ $T$ gates to implement a $k$-controlled Toffoli gate using $k-2$ unrestricted ancillae featured in \cite{ar:amm} without resorting to a computer optimization.  This is done via the use of the relative phase Toffoli and Toffoli-$V$ pair from Figure \ref{fig:rt}, dashed, and (\ref{circ:tof})--dashed, over the construction reported in Corollary \ref{cor2}.  We illustrate how this works using the circuit from Figure \ref{circ:tofn-3}(c) and observe that the arguments easily generalize to arbitrary $k$.  Substituting relative phase and special relative phase-$V$ pair implementations into the construction in Figure \ref{circ:tofn-3}(c) replaces each relative phase Toffoli with a circuit containing $4$ $T$ gates.  The total number of the $T$ gates would thus be $48$ (for arbitrary $k$, $16k-32$), higher than $40$ \cite{ar:amm}.  However, observe that $R_3TOF$ and $R_3^{-1}TOF$ are inverses of each other.  This means that the gates $T^\dagger$ and $H$ on the target qubit that the $R_3TOF$ ends with would cancel with $H$ and $T$ that the $R_3^{-1}TOF$ begins with.  This cancellation happens between all four such pairs $\{R_3TOF,R_3^{-1}TOF\}$ found in the circuit.  The total reduction is thus by 8 $T$ gates (for arbitrary $k$, $4k-12$), leading to a circuit with $40$ $T$ gates (for arbitrary $k$, $16k-32-4k+12=12k-20$).

\item {\it All ancillae in an unknown state, minimizing CNOT count.} \cite{www:q} includes an implementation where the controlled-controlled-$iX$ is used within \cite[Lemma 7.2]{j:bbcd} for all but two gates.  This construction relies on $36$ CNOT gates.  For arbitrary $n$, the CNOT count is $16n-44$, which we further refer to as cc-$iX$ implementation in Table \ref{tab:main}.  Observe that \cite{ar:mdmn} reports an implementation with $26$ two-qubit gates using two ancillae.  The optimization in \cite{ar:mdmn} is motivated by a computational model where the two-qubit interaction given by $diag\{I, X^{\pm t}\}$, where $t \in \mathbb{R}[0,1]$ and $X$ is Pauli-$X$, is tunable and parametrized by time.  Therefore, for example, a controlled-$\sqrt{NOT}$ would cost half as much as the CNOT, as it only needs to be evolved for half the time.  In our calculations given here, we do not allow such things to happen, but observe that it would be interesting to apply the reported relative phase Toffoli constructions within that framework.  Controlled-$\sqrt{NOT}$ may be implemented as a 2-CNOT circuit \cite[Figure 4.6]{bk:nc}.  The $26$ two-qubit gate circuit of \cite{ar:mdmn} has $18$ controlled-$\sqrt{NOT}$ gates and $8$ CNOT gates, therefore it would be transformed into one with $44$ CNOT gates.  Note, however, that it would make little sense from the point of view of the computational model considered in \cite{ar:mdmn}, as a length-$0.5$ interaction is being replaced with a length-$2$ interaction. 

In comparison, our solution, given by (\ref{circ:tof5any}), is a circuit with $20$ $(=6+4+6+4)$ CNOT gates that uses only one ancilla---latter being provably optimal within the framework of Clifford+$T$ circuits.  
\end{itemize}

We generalize the above examples of Toffoli-4 and Toffoli-5 optimization to any number of qubits in the following two Propositions.

\begin{lemm}\label{lem:tof0}
A size $n \geq 4$ multiple control Toffoli gate $TOF^n$ may be implemented using $\lceil \frac{n-3}{2} \rceil$ ancillary qubits, set to and returned in the value $\ket{0}$, by a circuit with:
\begin{itemize}
\item $8n-17$ $T$ gates,
\item $6n-12$ CNOT gates, and
\item $4n-10$ Hadamard gates.
\end{itemize}
\end{lemm}
\begin{proof}
The proof is by induction. The statement is clearly true for $n=4$ and $n=5$, as has been explicitly verified in the previous discussions.  To prove the transition from an even $n=2k$ to the odd $n=2k+1$ observe that the middle gate $TOF^3$ can be replaced with the circuit (\ref{circ:tof4zero}).  This introduces an $RTOF^3$, Figure \ref{fig:rt}, dashed, and its inverse. Note that a new ancillary qubit is being introduced on this step, and the gate counts increase by $8=4+4$ for $T$, by $6=3+3$ for CNOT, and by $4=2+2$ for Hadamard.  The transition from an odd $n=2k+1$ to the even $n=2k+2$ is accomplished via replacing $RTOF^3$ with $RTOF^4$, Figure \ref{circ:rtof4}, and its inverse with the inverse of $RTOF^4$. Observe that the gate counts grow by $8/6/4$ for $T$/CNOT/Hadamard, but no new ancilla is being introduced. 
\end{proof}

Note that \cite{j:s} reports a circuit with $n-3$ $\ket{0}$ ancillae, $8n-17$ $T$ gates, $8n-18$ CNOT gates, and $4n-10$ Hadamard gates.

\begin{lemm}\label{lem:tofx}
A size $n \geq 5$ multiple control Toffoli gate $TOF^n$ may be implemented by a circuit using $\lceil \frac{n-3}{2} \rceil$ ancillary qubits residing in an arbitrary state and returned unchanged, by a circuit with:
\begin{itemize}
\item $8n-16$ $T$ gates,
\item $8n-20$ CNOT gates, and
\item $4n-10$ Hadamard gates.
\end{itemize}
\end{lemm}
\begin{proof}
To assist with proving this Proposition, define the following gates:
\begin{enumerate}
\item $RTL(a,b,c)$ per Figure \ref{fig:rt}, dashed. This is a relative phase Toffoli gate. The implementation contains $9$ elementary gates: $4$ $T$ gates, $3$ CNOTs, and $2$ Hadamards. 
\item $RTS(a,b,c)$ per Figure \ref{fig:rt}, gates 2-6.  This is a relative phase Toffoli followed by a $V(b,c)$ that removes the last four gates.  The circuit contains $5$ elementary gates: $2$ $T$ gates, $2$ CNOTs, and $1$ Hadamard.
\item $SRTS(a,b,c)$ per circuit (\ref{circ:tof}), dashed. This is a Toffoli gate (as such it is also a type-$\{b\}$ special form relative phase Toffoli) followed by a $V(a,c)$ that removes last six gates. $SRTS$ contains $9$ elementary gates: $4$ $T$ gates, $4$ CNOTs, and $1$ Hadamard. 
\item $RT4L(a,b,c,d)$ per Figure \ref{circ:rtof4}. This is a $4$-qubit relative phase Toffoli.  It contains $8$ $T$ gates, $6$ CNOTs, and $4$ Hadamards.
\item $RT4S(a,b,c,d)$ per Figure \ref{circ:rtof4}, dashed. This is a relative phase Toffoli-4 $RT4L(a,b,c,d)$ followed by a $V(b,c,d)$ that removes last $8$ gates.  It is composed of the following elementary gates: $4$ $T$ gates, $4$ CNOTs, and $2$ Hadamards.
\end{enumerate}
We first prove the Proposition for the resource count of $n-3$ ancillae, $8n-16$ $T$ gates, $8n-18$ CNOT gates, and $4n-10$ Hadamard gates, and then introduce the $RT4L/RT4S$ gates that further improve the ancilla and CNOT count.  The proof relies on the construction found in Figure \ref{circ:tofn-3}(c).  Assuming qubits are numbered $1$ to $2n-3$ and we are attempting to implement $TOF^n(1,2,...,n-1;2n-3)$, select the gates in Figure \ref{circ:tofn-3}(c) as follows: 
\begin{enumerate}
\item\label{i1} First gate is $SRTS(n-1,2n-4,2n-3)$.
\item\label{i2} Next $k=1..n-4$ gates are $RTS(2n-4-k,n-1-k,2n-3-k)$.
\item\label{i3} Next gate is $RTL(1,2,n)$.
\item\label{i4} Next $k=1..n-4$ gates are $RTS^{-1}(n-1+k,k+2,n+k)$ (inverses of the gates in item \ref{i2} read in reverse order).
\item\label{i5} Next gate is $SRTS^{-1}(n-1,2n-4,2n-3)$ (this is the matching inverse pair for the gate in item \ref{i1}).
\item\label{i6} Next $k=1..n-4$ gates are $RTS(2n-4-k,n-1-k,2n-3-k)$ (same as item \ref{i2}).
\item\label{i7} Next gate is $RTL^{-1}(1,2,n)$ (this is the matching inverse for the gate in item \ref{i3}).
\item\label{i8} Last $k=1..n-4$ gates are $RTS^{-1}(n-1+k,k+2,n+k)$ (same as item \ref{i4}).
\end{enumerate}
Observe that the desired preliminary gate counts are satisfied.  Next step is introducing $RT4L/RT4S$ gates to replace as many $RTL$ and $RTS$ as possible. 
\begin{enumerate}
\item First, replace the circuit $RTS(n,3,n+1)RTL(1,2,n)RTS^{-1}(n,3,n+1)$ (last gate in item \ref{i2}, the gate in item \ref{i3}, and first gate in item \ref{i4}) with $RT4L(1,2,3,n+1)$ and $RTS(n,3,n+1)RTL^{-1}(1,2,n)$ $RTS^{-1}(n,3,n+1)$ (last gate in item \ref{i6}, the gate in item \ref{i7}, and first gate in item \ref{i8}) with $RT4L^{-1}(1,2,3,n+1)$.  Note that this procedure may only apply for $n \geq 5$.  It furthermore reduces the CNOT count from $7=2+3+2$ to $6$ twice, for a total saving of $2$ CNOTs.  Finally, observe that the qubit $n$ is no more used.  Thus, we save one ancillary qubit worth of computational space.
\item For $k=1..\lceil\frac{n-6}{2}\rceil$ we introduce four $RT4S$ gates by replacing a pair of neighbouring $RTS$ on the left and right hand sides of the previous step. In particular, we replace $RTS(n+2k,2k+3,n+2k+1)RTS(n-1+2k,2k+2,n+2k)$ (item \ref{i2}) with $RT4S(n-1+2k,2k+2,2k+3,n+2k+1)$ and $RTS^{-1}(n-1+2k,2k+2,n+2k)RTS^{-1}(n+2k,2k+3,n+2k+1)$ (item \ref{i4}) with $RT4S^{-1}(n-1+2k,2k+2,2k+3,n+2k+1)$, and similarly in the second half of the circuit (items \ref{i6}, \ref{i8}).  Observe that this operation does not change the gate counts, but frees up qubit $n+2k$ that is no more used, providing a reduction of one ancilla.
\end{enumerate}
The total reductions from the above construction are a pair of CNOT gates, and $\lfloor \frac{n-3}{2} \rfloor$ qubits, leading to the resource counts as announced in the statement of the Proposition.

Looking at the following circuit helps visualize all replacements and gate counts:

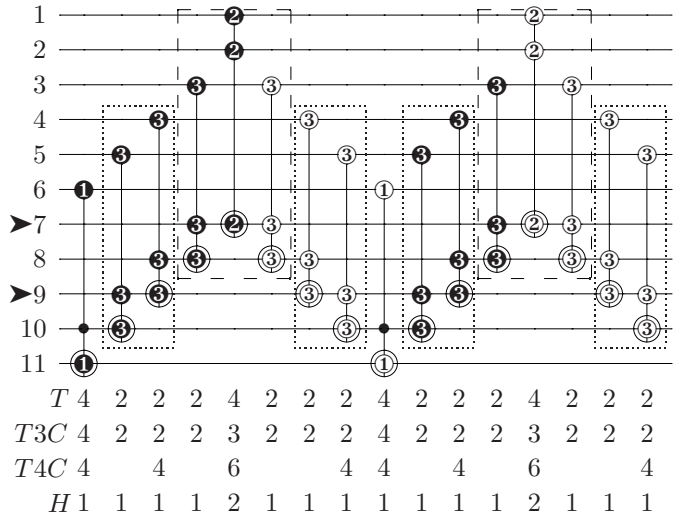
\begin{eqnarray*}\label{circ:tofn-32}
\Qcircuit @C=0.4em @R=.3em @!R {
\lstick{1}	& \qw 	& \qw 	& \qw 	& \qw 	& \ctrltwo{1} 	& \qw & \qw 	& \qw  		& \qw 		& \qw 		& \qw 		& \qw 		& \ctrlotwo{1} 	& \qw 		& \qw 		& \qw 		& \qw \\
\lstick{2}	& \qw & \qw & \qw 	& \qw & \ctrltwo{5} & \qw 		& \qw 	& \qw  		& \qw 		& \qw 		& \qw 		& \qw 		& \ctrlotwo{5} 	& \qw 		& \qw 		& \qw 		& \qw \\
\lstick{3}	& \qw 	& \qw & \qw & \ctrlthree{4} & \qw 	& \ctrlothree{4} 	& \qw & \qw  & \qw & \qw & \qw 		& \ctrlthree{4} 	& \qw 		& \ctrlothree{4} 	& \qw 		& \qw 		& \qw \\
\lstick{4}	& \qw 	& \qw & \ctrlthree{4} & \qw & \qw 	& \qw 	& \ctrlothree{4} & \qw  & \qw 	& \qw 	& \ctrlthree{4} & \qw 		& \qw 		& \qw 		& \ctrlothree{4} 	& \qw 		& \qw \\
\lstick{5}	& \qw 	& \ctrlthree{4} & \qw 	& \qw 	& \qw & \qw & \qw 	& \ctrlothree{4} & \qw 	& \ctrlthree{4}	& \qw 	& \qw 	& \qw 		& \qw 		& \qw 	 	& \ctrlothree{4}	& \qw \\
\lstick{6}	& \ctrlone{4}& \qw	& \qw 	& \qw 	& \qw 	& \qw 	& \qw 	& \qw	 & \ctrloone{4}	& \qw 		& \qw 		& \qw 		& \qw 		& \qw 		& \qw 	 	& \qw 		& \qw \\
\lstick{\text{\ding{228}} 7}& \qw & \qw & \qw & \ctrlthree{1} 	& \targtwo & \ctrlothree{1} & \qw & \qw  & \qw & \qw & \qw & \ctrlthree{1} & \targotwo & \ctrlothree{1} & \qw & \qw & \qw \\
\lstick{8}	& \qw & \qw & \ctrlthree{1} & \targthree	& \qw & \targothree	& \ctrlothree{1} & \qw  & \qw & \qw & \ctrlthree{1} & \targthree	& \qw & \targothree	& \ctrlothree{1} & \qw 	& \qw \\
\lstick{\text{\ding{228}} 9}& \qw & \ctrlthree{1} & \targthree	& \qw & \qw & \qw & \targothree	& \ctrlothree{1} & \qw & \ctrlthree{1}	& \targthree & \qw & \qw & \qw & \targothree & \ctrlothree{1}	& \qw \\
\lstick{10}	& \ctrl{1}	& \targthree & \qw & \qw & \qw & \qw & \qw & \targothree & \ctrl{1}	& \targthree	& \qw 		& \qw 		& \qw 		& \qw 	& \qw 	& \targothree	& \qw \\
\lstick{11}	& \targone	& \qw	& \qw 	& \qw 	& \qw & \qw & \qw & \qw	 	& \targoone	& \qw & \qw & \qw 		& \qw 		& \qw 		& \qw 	 	& \qw 		& \qw \\
T 						& 4 & 2 & 2 & 2 & 4 & 2 & 2 & 2 & 4 & 2 & 2 & 2 & 4 & 2 & 2 & 2 \\
\hspace{-5mm}T3C 		& 4 & 2 & 2 & 2 & 3 & 2 & 2 & 2 & 4 & 2 & 2 & 2 & 3 & 2 & 2 & 2 \\
\hspace{-5mm}T4C 		& 4 &   & 4 &   & 6 &   &   & 4 & 4 &   & 4 &   & 6 &   &   & 4 \\
H 						& 1 & 1 & 1 & 1 & 2 & 1 & 1 & 1 & 1 & 1 & 1 & 1 & 2 & 1 & 1 & 1 
\gategroup{1}{5}{8}{7}{0.4em}{--} \gategroup{1}{13}{8}{15}{0.4em}{--} \gategroup{4}{3}{10}{4}{0.4em}{.} \gategroup{4}{8}{10}{9}{0.4em}{.} \gategroup{4}{11}{10}{12}{0.4em}{.} \gategroup{4}{16}{10}{17}{0.4em}{.}
}
\end{eqnarray*}
In the above, dashed gates are replaced with $RT4L(1,2,3,8)$ and its inverse, freeing qubit $7$, and dotted gates are replaced with $RT4S(8,4,5,10)$ and its inverse, freeing qubit $9$.  Line starting with ``$T$'' reports the $T$ count,  line starting with ``$T3C$'' reports the $CNOT$ count when only $RTOF^3$ are being used, line starting with ``$T4C$'' reports the $CNOT$ count when $RTOF^4$ are allowed, and line starting with ``$H$'' reports the Hadamard gate count.
\end{proof}

We summarize the results in Table \ref{tab:main} and compare them against best known.  The names of the columns are self-explanatory.  Observe that \cite{ar:mdmn} features multiple control Toffoli implementations using $12n-34$ two-qubit gates over a circuit with $n-3$ ancillae.  In comparison, our implementation uses $8n-20$ CNOT gates over a circuit with only $\lceil \frac{n-3}{2} \rceil$ ancillae.  It is furthermore interesting to highlight that in terms of implementing a multiple control Toffoli gate the cost of moving away from using unrestricted ancillae to ancillae residing in the state $\ket{0}$ is only one $T$ gate, but in terms of the CNOTs, it is a noticeable term, $2n-8$. 

\begin{table*}[ht]
\begin{tabular}{ccccccccc} \hline \hline
Gate 	& Source 	 		& Optimization goal & \# $T$ 	& \# CNOT	& \# H 	& \# P/Z & \# ancillae 	& Ancillae type \\ \hline
$TOF^4$ & \cite{ar:amm} 			& $T$       & 15 	 	& 35 		& 6  	& 3 	& 1  			& $\ket{0}$ \\  
		& \cite{j:s}				& $T$		& 15 		& 14 		& 6		& 0 	& 1  			& $\ket{0}$ \\
		& Ours 						& $T$, CNOT & 15		& 12		& 6 	& 0 	& 1  			& $\ket{0}$ \\
		& \cite{ar:amm} 			& $T$       & 16		& 54 		& 6 	& 6		& 1  			& $\ket{x}$ \\
		& cc-$iX$ \cite{www:q} 		& CNOT	    & 22		& 20 		& 8 	& 0		& 1  			& $\ket{x}$ \\
		& Ours 						& $T$, CNOT & 16 		& 14 		& 6 	& 0 	& 1				& $\ket{x}$ \\ \hline
$TOF^5$ & \cite{ar:amm} 			& $T$		& 23		& 63		& 10	& 6 	& 2				& $\ket{00}$ \\
		& \cite{j:s}				& $T$		& 23		& 22		& 10 	& 0		& 2				& $\ket{00}$ \\
		& Ours						& $T$, CNOT & 23		& 18		& 10 	& 0		& 1				& $\ket{0}$ \\
		& \cite{ar:amm} 			& $T$ 		& 28		& 90 		& 10 	& 13	& 2 			& $\ket{xx}$ \\
		& cc-$iX$ \cite{www:q}		& CNOT	 	& 38		& 36		& 16 	& 0 	& 2 			& $\ket{xx}$ \\ 
		& Ours 						& $T$, CNOT & 24		& 20		& 10	& 0 	& 1				& $\ket{x}$ \\ \hline
$TOF^6$ & \cite{ar:amm} 			& $T$       & 31 		& 94 		& 14 	& 9 	& 3 			& $\ket{000}$ \\ 
		& \cite{j:s}				& $T$		& 31		& 30		& 14 	& 0		& 3 			& $\ket{000}$ \\ 
		& Ours						& $T$, CNOT & 31		& 24		& 14 	& 0		& 2 			& $\ket{00}$ \\ 
		& \cite{ar:amm} 			& $T$       & 40		& 132		& 14 	& 20	& 3 			& $\ket{xxx}$ \\ 
		& cc-$iX$ \cite{www:q}		& CNOT		& 46		& 52 		& 24 	& 0		& 3 			& $\ket{xxx}$ \\ 
		& Ours						& $T$, CNOT & 32		& 28		& 14 	& 0 	& 2 			& $\ket{xx}$ \\ \hline
$TOF^{11}$& \cite{ar:amm} 			& $T$     	& 71 		& 232 		& 34 	& 24 	& 8 			& $\ket{00000000}$ \\
		& \cite{j:s}				& $T$		& 71 		& 70		& 34 	& 0 	& 8 			& $\ket{00000000}$ \\
		& Ours						& $T$, CNOT & 71		& 54		& 34 	& 0 	& 4 			& $\ket{0000}$ \\
		& \cite{ar:amm} 			& $T$       & 100		& 328 		& 34 	& 55	& 8 			& $\ket{xxxxxxxx}$ \\
		& cc-$iX$ \cite{www:q}		& CNOT		& 134 		& 132 		& 64 	& 0		& 8 			& $\ket{xxxxxxxx}$ \\
		& Ours						& $T$, CNOT & 72		& 68 		& 34 	& 0 	& 4 			& $\ket{xxxx}$ \\ \hline
$TOF^n, n \geq 5$ & \cite{ar:amm} 	& $T$       & 8n-17 	& N/A		& N/A	& N/A	& n-3			& $\ket{00...0}$ \\
		& \cite{j:s}				& $T$		& 8n-17 	& 8n-18 	& 4n-10 & 0		& n-3			& $\ket{00...0}$ \\
		& Ours						& $T$, CNOT & 8n-17 	& 6n-12 	& 4n-10 & 0 	& $\lceil \frac{n-3}{2} \rceil$ & $\ket{00...0}$ \\
		& \cite{ar:amm}				& $T$ 		& 12n-32 	& N/A		& N/A	& N/A	& n-3			& $\ket{xx...x}$ \\
		& cc-$iX$ \cite{www:q}		& CNOT		& 16n-42 	& 16n-44 	& 8n-24 & 0 	& n-3			& $\ket{xx...x}$ \\
		& Ours						& $T$, CNOT & 8n-16 	& 8n-20 	& 4n-10 & 0		& $\lceil \frac{n-3}{2} \rceil$ & $\ket{xx...x}$ \\ \hline
\end{tabular} 
\caption{Optimization of the multiple control Toffoli gates using $RTOF^3$ and $RTOF^4$ gates.} \label{tab:main}
\end{table*}

\section{Open problems}

The problem of systematically synthesizing and analyzing multiple control relative phase Toffoli implementations---both unrestricted as well as the special form, is important to address next.  The results of such a search could be used directly to optimize implementations of the multiple control Toffoli gates, arithmetic parts of quantum algorithms, and reversible circuits.  

How efficient may a relative phase multiple control Toffoli gate implementation be?  In the $3$-qubit case the answer is: it requires at least $3$ CNOTs as a circuit over CNOT and any single-qubit gates library, as otherwise, per Corollary \ref{cor:1}, we would come to a contradiction with any lower CNOT gate count \cite{ar:sm}.  If it is established that the Toffoli gate requires $7$ $T$ gates in the presence of ancillae, a similar argument can be applied towards showing that any relative phase Toffoli gate requires at least $4$ $T$ gates as a circuit over Clifford+$T$ library. 
 
The reported constructions obtain best solutions simultaneously for two circuit cost metrics arising from different considerations, the CNOT-count and the $T$-count.  It may be that this is not a coincidence.  Is there a relation between these two resource counts?

\section{Conclusion}

In this paper, we reported an approach for systematic optimization of quantum circuits via replacing suitable pairs of the multiple control Toffoli gates with their relative phase implementations.  This operation preserves the functional correctness.  However, since the relative phase Toffolis are easier to implement than their regular counterparts, the advantage can be witnessed through the optimized resource counts.  We have furthermore illustrated the advantage via optimizing and, when applicable, explaining the nature of best known implementations of the multiple control Toffoli gates.  Our demonstrated optimizations include a simultaneous optimization of the $T$ count by a factor of $\frac{4}{3}$ in the leading constant, the CNOT count by a factor $2$ in the leading constant, and the number of ancillary qubits by a factor of $2$ in the leading constant.  The above refers to the optimization of the circuit implementing the multiple control Toffoli gate using arbitrary ancillae, whose construction resulted from employing the relative phase Toffoli gates.

\section*{Acknowledgements}

I wish to thank anonymous reviewers for their useful comments. 

Circuit diagrams were drawn using qcircuit.tex package, \href{http://physics.unm.edu/CQuIC/Qcircuit/}{http://physics.unm.edu/CQuIC/Qcircuit/}.

This material was based on work supported by the National Science Foundation, while working at the Foundation. Any opinion, finding, and conclusions or recommendations expressed in this material are those of the author and do not necessarily reflect the views of the National Science Foundation.

%\section*{Appendix} \label{www:nb}
%
%This appendix contains Mathematica code that can be helpful in proving Lemmas \ref{lemma:main1} and \ref{lemma:main2}.
%
%\begin{verbatim}
%
%
%\end{verbatim}


\begin{thebibliography}{2}

\bibliographystyle{alpha}

\bibitem{ar:mk} C. Monroe and J. Kim, Science {\bf 339}, 1164--1169 (2013).

\bibitem{ar:dh} M. H. Devoret and R. J. Schoelkopf, Science {\bf 339}, 1169--1174 (2013).

\bibitem{j:bbcd} A. Barenco, C. H. Bennett, R. Cleve, D. P. DiVincenzo, N. Margolus, P. Shor, T. Sleator, J. Smolin, and H. Weinfurter, Phys. Rev. A {\bf 52}, 3457--3467 (1995). 

\bibitem{bk:nc} M. A. Nielsen and I. L. Chuang, {\em Quantum Computation and Quantum Information}, Cambridge University Press, New York (2000).

\bibitem{www:nb} Helpful Mathematica calculations are \href{http://www.umiacs.umd.edu/~dmaslov/papers/mathematicacomputations.txt}{available online} at 
http://www.umiacs.umd.edu/\textasciitilde dmaslov/papers/mathe maticacomputations.txt. 

\bibitem{c:ccdf} A. M. Childs, R. Cleve, E. Deotto, E. Farhi, S. Gutmann, and D. A. Spielman, Proc. 35th ACM STOC, 59--68 (2003).

\bibitem{j:s} P. Selinger, Phys. Rev. A {\bf 87}, 042302 (2013).

\bibitem{ar:bk} S. Bravyi and A. Kitaev, Phys. Rev. A {\bf 71}, 022316 (2005).

\bibitem{ar:somo} A. Sorensen and K. Molmer, Phys. Rev. Lett. {\bf 82}, 1971--1974 (1999).

\bibitem{ar:amm} M. Amy, D. Maslov, and M. Mosca, IEEE Trans. CAD {\bf 33}(10), 1476--1489 (2014).

\bibitem{ar:sk} G. Song and A. Klappenecker, Quantum Information and Computation {\bf 4}, 361--372 (2004).

\bibitem{www:q} \href{http://www.mathstat.dal.ca/~selinger/quipper/doc/frames.html}{Quipper 0.5}, http://www.mathstat.dal.ca/\textasciitilde selinger/ quipper/doc/frames.html, released September 2013.

\bibitem{ar:sm} V. V. Shende and I. L. Markov, Quantum Information and Computation {\bf 9}(5-6), 461--486, (2009).

\bibitem{ar:mdmn} D. Maslov, G. W. Dueck, D. M. Miller, and C. Negrevergne, IEEE Trans. CAD {\bf 27}(3), 436--444 (2008).


%\bibitem{ar:hrb} H. Haeffner, C.F. Roos, and R. Blatt.
%\emph{Quantum computing with trapped ions}, Physics Reports 469(4):155–-203, 2008,
%\href{http://arxiv.org/abs/0809.4368}{arXiv:0809.4368}.


\end{thebibliography}
\end{document}